\documentclass[10pt]{article}

\usepackage{geometry,amsthm,amssymb,amsmath,amsfonts,rotate,lmodern,amsmath,ifthen,eurosym,graphicx,algorithmic,algorithm,aliascnt,stmaryrd,subcaption,wrapfig,cite,latexsym,amssymb,url}
\usepackage[svgnames]{xcolor}
\usepackage{tikz}
\usepackage{array}
\usetikzlibrary{arrows.meta,patterns}

\RequirePackage[T1]{fontenc}
\RequirePackage[utf8]{inputenc}
\usepackage[english]{babel}
\usepackage{alphabeta}

\everymath{\color{MidnightBlack}}

\usepackage[pdftex,
backref=true,
plainpages = true,
pdfpagelabels,
hyperfootnotes=true,
pdfpagemode=FullScreen,
bookmarks=true,
bookmarksopen = true,
bookmarksnumbered = true,
breaklinks = true,
hyperfigures,
urlcolor = magenta,
urlcolor = MidnightBlack,
anchorcolor = green,
hyperindex = true,
colorlinks = true,
linkcolor = black!30!blue,
citecolor = black!30!green
]{hyperref}
\usepackage{cleveref}

\DeclareUnicodeCharacter{2192}{\ifmmode\to\else\textrightarrow\fi}
\DeclareUnicodeCharacter{2203}{\ensuremath\exists}
\DeclareUnicodeCharacter{183}{\cdot}
\DeclareUnicodeCharacter{2200}{\forall}
\DeclareUnicodeCharacter{2264}{\leq}
\DeclareUnicodeCharacter{2265}{\geq}
\DeclareUnicodeCharacter{8614}{\mathbin{\mapsto}}
\DeclareUnicodeCharacter{8656}{\Leftarrow}
\DeclareUnicodeCharacter{8657}{\Uparrow}
\DeclareUnicodeCharacter{8658}{\Rightarrow}
\DeclareUnicodeCharacter{8659}{\Downarrow}
\DeclareUnicodeCharacter{8669}{\rightsquigarrow}
\newcommand{\eqdef}{\stackrel{{\scriptsize\rm def}}{=}}
\DeclareUnicodeCharacter{8797}{\eqdef}
\DeclareUnicodeCharacter{8870}{\vdash}
\DeclareUnicodeCharacter{8873}{\Vdash}
\DeclareUnicodeCharacter{22A7}{\models}
\DeclareUnicodeCharacter{9121}{\lceil}
\DeclareUnicodeCharacter{9123}{\lfloor}
\DeclareUnicodeCharacter{9124}{\rceil}
\DeclareUnicodeCharacter{222A}{\cup}

\DeclareUnicodeCharacter{2208}{\in}
\DeclareUnicodeCharacter{9126}{\rfloor}
\DeclareUnicodeCharacter{2286}{\subseteq}
\DeclareUnicodeCharacter{9655}{\triangleright}
\DeclareUnicodeCharacter{9665}{\triangleleft}
\DeclareUnicodeCharacter{9671}{\diamond}
\DeclareUnicodeCharacter{9675}{\circ}
\DeclareUnicodeCharacter{10178}{\bot}
\DeclareUnicodeCharacter{10214}{} 
\DeclareUnicodeCharacter{10215}{} 
\DeclareUnicodeCharacter{10229}{\longleftarrow}
\DeclareUnicodeCharacter{10230}{\longrightarrow}
\DeclareUnicodeCharacter{10231}{\longleftrightarrow}
\DeclareUnicodeCharacter{10232}{\Longleftarrow}
\DeclareUnicodeCharacter{10233}{\Longrightarrow}
\DeclareUnicodeCharacter{10234}{\Longleftrightarrow}
\DeclareUnicodeCharacter{10236}{\longmapsto}
\DeclareUnicodeCharacter{10238}{\Longmapsto} 
\DeclareUnicodeCharacter{10503}{\Mapsto}    
\DeclareUnicodeCharacter{10971}{\mathrel{\not\hspace{-0.2em}\cap}}
\DeclareUnicodeCharacter{65294}{\ldotp}
\DeclareUnicodeCharacter{65372}{\mid}
\usepackage{color}
\definecolor{MidnightBlack}{rgb}{0.1,0.1,0.28}
\definecolor{MidnightBlue}{rgb}{0.1,0.1,0.44}
\definecolor{Black}{rgb}{0,0, 0}
\definecolor{Blue}{rgb}{0, 0 ,1}
\definecolor{Red}{rgb}{1, 0 ,0}
\definecolor{White}{rgb}{1, 1, 1}
\definecolor{Grey}{rgb}{.6, .6, .6}
\definecolor{Mygreen}{rgb}{.0, .7, .0}
\definecolor{Yellow}{rgb}{.55,.55,0}
\definecolor{Mustard}{rgb}{1.0, 0.86, 0.35}
\definecolor{applegreen}{rgb}{0.55, 0.71, 0.0}
\definecolor{darkturquoise}{rgb}{0.0, 0.81, 0.82}
\definecolor{celestialblue}{rgb}{0.29, 0.59, 0.82}
\definecolor{green_yellow}{rgb}{0.68, 1.0, 0.18}
\definecolor{crimsonglory}{rgb}{0.75, 0.0, 0.2}
\definecolor{darkmagenta}{rgb}{0.30, 0.0, 0.30}
\definecolor{internationalorange}{rgb}{1.0, 0.31, 0.0}
\definecolor{darkorange}{rgb}{1.0, 0.55, 0.0}

\newcommand{\rred}[1]{{#1}}

\newcommand{\abs}[1]{|#1|}
\newcommand{\edeg}{{\sf edeg}}

\newcommand{\yes}{{\sf yes}}

\newcommand{\remove}[1]{}
\newcommand{\removed}[1]{#1}

\newcommand{\cupall}{\pmb{\pmb{\bigcup}}}

\newcommand{\R}{\mathbb{R}}      
 
\newcommand{\opt}{\mathsf{opt}}

\newcommand{\calC}{\mathcal{C}} 
\newcommand{\ext}{\mathsf{ext}} 
 
\newcommand{\F}{\mathcal{F}}

\makeatletter
\newcommand{\newreptheorem}[2]{\newtheorem*{rep@#1}{\rep@title}
\newenvironment{rep#1}[1]{\def\rep@title{#2 \ref*{##1}}\begin{rep@#1}}{\end{rep@#1}}}
\makeatother

\newreptheorem{theorem}{Theorem}

\newcommand{\layer}{$\alpha$-thin layer\xspace}

\newcommand{\VC}{\textsc{Weighted Vertex Cover}\xspace}
\newcommand{\FVS}{\textsc{Weighted Feedback Vertex Set}\xspace}
\newcommand{\DIA}{\textsc{Weighted Diamond Hitting Set}\xspace}

\newcommand{\pumpkin}{\textsc{Weighted $c$-Bond Cover}\xspace}
\newcommand{\upumpkin}{\textsc{$c$-Bond Cover}\xspace}

\newcommand{\FD}{$\F$-\textsc{Vertex Deletion}\xspace}
\newcommand{\WFD}{\textsc{Weighted} $\mathcal{F}$-\textsc{Vertex Deletion}\xspace}

\newcounter{func}
\newcommand{\newfun}[1]{f_{\refstepcounter{func}\label{#1}\thefunc}}
\newcommand{\funref}[1]{\hyperref[#1]{f_{\ref*{#1}}}} 

\newcounter{con}

\newcommand{\conref}[1]{\hyperref[#1]{c_{\ref*{#1}}}} 

\newcommand{\mynewtheorem}[2]{
	\newaliascnt{#1}{dummy}
	\newtheorem{#1}[#1]{#2}
	\aliascntresetthe{#1}
}

\theoremstyle{plain}
\mynewtheorem{theorem}{Theorem}
\mynewtheorem{proposition}{Proposition}
\mynewtheorem{corollary}{Corollary}
\mynewtheorem{lemma}{Lemma}

\theoremstyle{definition}
\mynewtheorem{definition}{Definition}
\theoremstyle{remark}
\mynewtheorem{sublemma}{Sublemma}
\mynewtheorem{claim}{Claim}
\mynewtheorem{remark}{Remark}
\mynewtheorem{fact}{Fact}
\mynewtheorem{conjecture}{Conjecture}
\mynewtheorem{question}{Question}
\mynewtheorem{observation}{Observation}
\mynewtheorem{problem}{Problem}
\mynewtheorem{note}{Note}

\usepackage{titling}
\thanksmarkseries{arabic}

\newenvironment{proofofclaim}{\begin{list}{}{
              \setlength{\leftmargin}{0mm}
              } \item {\it Proof of Claim.}}{\hfill$\blacklozenge$\end{list}\medskip}

\usepackage{xspace}

\date{\empty}

\begin{document}

\title{A Constant-factor Approximation for \\ Weighted Bond Cover\thanks{The first author was supported by the ANR projects ASSK (ANR-18-CE40-0025) and ESIGMA (ANR-17-CE23-0010). Part of this work was done when the second author was a postdoc at New York University and supported by Simons Collaboration on Algorithms and Geometry. The last author was supported  by   the ANR projects DEMOGRAPH (ANR-16-CE40-0028), ESIGMA (ANR-17-CE23-0010), and the French-German Collaboration ANR/DFG Project UTMA (ANR-20-CE92-0027). Email addresses: \mbox{\sf eun-jung.kim@dauphine.fr}, \mbox{\sf euiwoong@umich.edu}, \mbox{\sf sedthilk@thilikos.info}\,.}}
\author{Eun Jung Kim\thanks{Université Paris-Dauphine, PSL University, CNRS, LAMSADE, 75016, Paris, France.}\and Euiwoong Lee\thanks{University of Michigan, Ann Arbor, USA.}\and Dimitrios M. Thilikos\thanks{LIRMM, Univ Montpellier, CNRS, Montpellier, France.}}
\maketitle

\begin{abstract}

\noindent The \WFD\ for a class $\F$ of graphs asks, weighted graph $G$, for a minimum weight vertex set $S$ such that $G-S\in\F.$ 
The case when $\F$ is minor-closed 
and excludes some graph as a minor has received particular attention but a constant-factor approximation remained elusive for \WFD. 
Only three cases of minor-closed $\F$ are known to admit constant-factor approximations, namely \textsc{Vertex Cover}, 
\textsc{Feedback Vertex Set} and \textsc{Diamond Hitting Set}. 
We study the problem  for the class $\F$ of $\theta_c$-minor-free graphs, under the equivalent  setting of the \pumpkin\ problem, and present a constant-factor approximation algorithm using the primal-dual method. For this, we leverage
a structure theorem implicit in [{\em Joret, Paul, Sau, Saurabh, and Thomass\'{e},  SIDMA'14}] which states the following: any graph $G$ containing a $\theta_c$-minor-model 
either contains a large two-terminal {\sl protrusion}, or contains a constant-size $\theta_c$-minor-model, or a collection of pairwise disjoint  {\sl constant-sized} connected sets that can be contracted simultaneously 
to yield a dense graph. In the first case, we tame the  graph 
by replacing the protrusion with a special-purpose weighted gadget. For the second and  third case, we  provide a weighting scheme which guarantees a local approximation ratio. 
%
%
Besides making an important step in the quest of (dis)proving  a constant-factor approximation for \WFD, 
our result may be useful as a template for   algorithms for other minor-closed families.
\end{abstract}

\noindent{\bf keywords}:  Constant-factor approximation algorithms, Primal-dual method, Bonds in graphs,  Graph minors, Graph modification problems.

%
%

\section{Introduction}

For a class $\F$ of graphs, the problem \WFD\ asks, given weighted graph $G=(V,E,w),$  for  a vertex set $S\subseteq V$   
of minimum weight such that $G-S$ belongs to the class $\F.$ The \WFD\ captures classic graph problems such as \VC\ and \FVS, 
which corresponds to $\F$ being the classes of edgeless and acyclic graphs, respectively. 
A vast literature is devoted to the study of ({\sc Weighted}) \FD\ for various instantiations of $\F,$ both  in approximation algorithms and in parameterized complexity. 
Much of the work considers  a class $\F$ that is characterized by a set of forbidden (induced) subgraphs
~\cite{reed2004finding,kratsch2014compression,cao2016linear,agrawal2016faster,
jansen2017approximation,agrawal2018feedback,agrawal2019interval,AEKO2019,Ahn0L20,LMPPS2020} 
or that is minor-closed
~\cite{FominL0Z20approx,ChudakGHW98,fiorini2010hitting,fomin2012planar,0002LPRRSS16,FominLMPS16,KawarabayashiS17,bansal2017lp,agrawal2018polylogarithmic,gupta2019losing,
FominLP0Z20,SauST20,BasteST20}, 
thus characterized by a (finite) set of forbidden minors. 

Lewis and Yannakakis~\cite{LY1980} showed that \FD, the unweighted version of \WFD, is {\sf NP}-hard whenever $\F$ is nontrivial (there are infinitely many graphs in and outside 
of $\F$) and hereditary (is closed under taking induced subgraphs). It was also long known that \FD\ is {\sf APX}-hard for every non-trivial hereditary class $\F$~\cite{LundY93}. 
So, the natural question is for which class $\F,$ \FD\ and \WFD\ admit constant-factor approximation algorithms. 

When $\F$ is characterized by a finite set of forbidden induced subgraphs, a constant-factor approximation for \WFD\ is readily derived with LP-rounding technique. 
Lund and Yannakakis~\cite{LundY93} conjectured that for $\F$ characterized by a set of minimal forbidden induced subgraphs, 
the finiteness of $\F$ defines the borderline between approximability and inapproximability with constant ratio of \FD. 
This conjecture was refuted due to the existence of 2-approximation for \FVS~\cite{BafnaBF99,BECKER1996,ChudakGHW98}.  
Since then, a few more classes with an infinite set of forbidden induced subgraphs 
are known to allow constant-factor approximations for \FD, such as block graphs~\cite{agrawal2016faster},  
3-leaf power graphs~\cite{AEKO2019}, interval graphs~\cite{cao2016linear}, ptolemaic graphs~\cite{Ahn0L20}, and bounded treewidth graphs~\cite{fomin2012planar,gupta2019losing}. 
That is, we are only in the nascent stage when it comes to charting the landscape of ({\sc Weighted}) \FD\ as to constant-factor approximability. 
In the remainder of this section, we focus on the case where $\F$ is a minor-closed class.
\smallskip

\noindent{\bf Known results on ({\sc Weighted}) {\sc \FD}.} 
According to Robertson and Seymour theorem, every non-trivial minor-closed graph class $\F$ is characterized by a finite set, called ({\em minor}) {\em  obstruction set}, of minimal forbidden minors, called ({\em minor}) {\em obstructions}~\cite{RobertsonS04}. 
It is also well-known that $\F$ has bounded treewidth if and only if one of the obstructions is planar~\cite{RobertsonS86}.  
Therefore, the \FD\ for $\F$ excluding at least one planar graph as a minor can be deemed a natural extension of \textsc{Feedback Vertex Set}.
In this context, it is not surprising that \FD, for minor-closed $\F$, attracted particular attention in parameterized complexity, where \text{Feedback Vertex Set} 
was considered the flagship problem serving as an igniter and a testbed for  new techniques.  

For every minor-closed $\F,$ the class of \yes-instances to the decision version of \FD\ is minor-closed again (for every fixed size of a solution), thus there exists a finite obstruction set  for the set of its \yes-instances. With a minor-membership test algorithm~\cite{KAWARABAYASHI2012424}, this implies that \FD\ is fixed-parameter tractable. The caveat is, such a fixed-parameter algorithm is non-uniform and non-constructive, and the  
 exponential term in the running time is gigantic.  
Much endeavour was made to reduce the parametric dependence of  such algorithms for \FD.
The case when $\F$ has  bounded treewidth is now understood well. The corresponding \FD\ is known to be solvable in time $2^{{\cal O}_{\F}(k)}\cdot n^{{\cal O}(1)}$
~\cite{fomin2012planar,0002LPRRSS16} where $k$ denotes the size of the optimal solution, and the single-exponential dependency on $k$ is asymptotically optimal under the {\sl Exponential Time Hypothesis}\footnote{The {\sf ETH} states that 3-\textsc{SAT} on $n$ variables cannot be solved in time $2^{o(n)},$  see~\cite{ImpagliazzoPZ01which} for more details.}~{\cite{0002LPRRSS16}}. (See also \cite{SauST20} for recent parameterized algorithms for general minor-closed ${\cal F}$'s). 
%

Turning to approximability, the (unweighted) \FD\ can be approximated within a constant-factor when $\F$ has treewidth at most some constant $t$, 
or equivalently, when the obstruction set of ${\cal F}$ contains some planar graph. The first general result 
in this direction was the  randomized $f(t)$-approximation 
of Fomin, Lokshtanov, Misra, and Saurabh~\cite{fomin2012planar}
 Gupta, Lee, Li, Manurangsi, and W\l{}odarczyk~\cite{gupta2019losing} made a further progress 
with an ${\cal O}(\log t)$-approximation algorithm. Unfortunately, such approximation algorithms whose approximation ratio depends only on $\F$ are not known 
when the input is {\sl weighted}. A principal reason for this is that most of the techniques developed for the unweighted case do not extend to the weighted setting.
In this direction,
Agrawal, Lokshtanov, Misra, Saurabh, and Zehavi~\cite{agrawal2018polylogarithmic} 
presented a randomized ${\cal O}(\log^{1.5} n)$-approximation algorithm and a deterministic ${\cal O}(\log^2 n)$-approximation algorithm which 
run in time $n^{{\cal O}(t)}$ when $\F$ has treewidth at most $t.$ It is reported in~\cite{agrawal2018polylogarithmic} that an ${\cal O}(\log n\cdot \log \log n)$-approximation can be 
deduced from the approximation algorithm of Bansal, Reichman, and Umboh~\cite{bansal2017lp} for the edge deletion variant of \WFD. 
For the class $\F$ of planar graphs, Kawarabayashi and Sidiropoulos~\cite{KawarabayashiS17} presented an approximation algorithm for \FD\ with polylogarithmic approximation ratio 
running in quasi-polynomial time. Beyond this work, no nontrivial approximation algorithm is known for $\F$ of unbounded treewidth.

Regarding {\sl constant-factor} approximability for \WFD\  with minor-closed $\F,$ only three results are known till now. 
For the \VC, it was observed early that a 2-approximation can be instantly derived from the half-integrality of LP~\cite{NemhauserT74}. 
The local-ratio algorithm  by Bar-Yehuda and Even~\cite{BARYEHUDA198527} was presumably the first primal-dual algorithm 
and laid the groundwork for subsequent development of the primal-dual method.\footnote{In this paper, we consider local-ratio and primal-dual as the same algorithms design paradigm and use the word primal-dual throughout the paper even when the underlying LP is not explicitly given. 
We refer the reader to the classic survey of Bar-Yehuda Bendel, Freund, and Rawitz~\cite{BBFR04} for the equivalence. }
For the \FVS,  2-approximation algorithms were proposed using the primal-dual method~\cite{BafnaBF99,BECKER1996,ChudakGHW98}. Furthermore, 
a 9-factor approximation algorithm was given for \DIA\ by Fiorini, Joret, and Pietropaoli~\cite{fiorini2010hitting} in 2010. 
To the best of our knowledge, since then, no progress is done on approximation with constant ratio for minor-closed $\F.$
Table~\ref{table:summary} summarizes the state of art. 
\medskip

%

\begin{table}[h!]
\centering
\begin{tabular}{|c|c|c|}
\hline
 & Unweighted & Weighted \\ \hline
Forests & \multicolumn{2}{c|}{2~\cite{BafnaBF99,BECKER1996,ChudakGHW98}} \\ \hline
$\theta_3$-free & \multicolumn{2}{c|}{9~\cite{fiorini2010hitting}} \\ \hline
$\theta_c$-free & $O(\log c)$~\cite{gupta2019losing} & $O_c(1)$ (this work) \\ \hline
Treewidth $\leq t$ & $O(\log t)$~\cite{gupta2019losing} & $O(\log^{1.5} n)$~\cite{agrawal2018polylogarithmic} \\ \hline
Planar  & $\log^{O(1)} n$~\cite{KawarabayashiS17} & Not studied \\ \hline
\end{tabular}
\caption{Current best approximation ratios for \FD\ for minor-closed ${\cal F}$. } 
\label{table:summary}
\end{table}

For minor-closed $\F$ with graphs of bounded treewidth, the known approximation algorithms for ({\sc Weighted}) \FD\   
take one of the  following  two avenues. 
First, the algorithms in~\cite{bansal2017lp,agrawal2018polylogarithmic,gupta2019losing} draw on the fact that a graph of constant treewidth 
has a constant-size separator which breaks down the graph into {\sl smaller} pieces. The measure for smallness is an important 
design feature of these algorithms. Regardless of the design specification, however, it seems there is an inherent bottleneck to 
extend these algorithmic strategy to handle weights while achieving a constant approximation ratio; 
the above results either use an algorithm for the \textsc{Balanced Separator} problem that does not admit a constant-factor approximation ratio, under the Small Set Expansion Hypothesis~\cite{LRV13}, 
or use a relationship between the size of the separator and the size of resulting pieces that do not hold for weighted graphs.

The second direction  is  the primal-dual method~\cite{BARYEHUDA198527,BafnaBF99,BECKER1996,ChudakGHW98,fiorini2010hitting}. 
The constant-factor approximation of~\cite{fomin2012planar} for \FD\ is also based on the same core observation of the primal-dual algorithm such as~\cite{BECKER1996}. 
The 2-approximation for \FVS\ became available by introducing a new LP formulation 
which translates the property `$G-X$ is a forest' in terms of the {\sl sum of degree contribution} of $X.$ 
The idea of expressing the sparsity condition of $G-X$ in terms of the degree contribution of $X$ again played the key role in~\cite{fiorini2010hitting} 
for \DIA. However,  the (extended) sparsity inequality of~\cite{fiorini2010hitting} is highly intricate as the LP constraint 
describes the precise structure of diamond-minor-free graphs (after {\sl taming} the graph via some special protrusion replacer). 
Therefore, expressing the sparsity condition for other classes $\F$ with tailor-made  LP constraints  
is likely to be prohibitively convoluted. This implies that a radical simplification of the known algorithm for, say, \DIA will 
be necessary if one intends to apply the primal-dual method for broader classes. 



\bigskip

\noindent{\bf Our result and the key ideas.} 
The central problem we study is the \WFD\ where $\F$ is the class of $\theta_c$-minor-free graphs: 
a weighted graph $G=(V,E,w)$ is given as input, and the goal is 
to find a vertex set $S$ of minimum weight 
such that $G-S$ is $\theta_c$-minor-free.\footnote{The graph $θ_{c}$ is the graph on two vertices joined by $c$ 
 parallel edges.} We call this particular  problem the \pumpkin\ problem, as we believe that this nomenclature is more adequate  for reasons to be clear later (see \autoref{label_rajeunissement} in Section~\ref{label_unsuspicious}).
Our main result is the following.

\begin{theorem}\label{label_resourcefulness}
There is a constant-factor approximation algorithm for \pumpkin\ which runs in time\footnote{We use notation ${\cal O}_{c}(f)$ in order to denote $g(c)\cdot {\cal O}(f),$ for some computable function $g:\Bbb{N}\to\Bbb{N}.$} $\mathcal{O}_c(n^{{\cal O}(1)}),$ for every positive $c.$ 
\end{theorem}

Let us briefly recall the classic 2-approximation algorithms for \FVS~\cite{BafnaBF99,BECKER1996}. 
These algorithms repeatedly delineate a vertex subset $S$ on which the induced subgraph contains 
an obstruction (a cycle), and ``peel off'' a weighted graph on $S$ from the current weighted graph so that the weight of at least one vertex of the current graph 
drops  to zero. The crux of this approach is to create a weighted graph to peel off (or {\sl design a weighting scheme}) on which {\sl every} (minimal) feasible solution is consistently an $\alpha$-approximate solution. 
We remark that peeling-off of a weighted graph on $S$   
can be viewed as increasing the dual variable (from zero) corresponding to $S$ until some dual constraint becomes tight, as articulated in~\cite{ChudakGHW98}.

If one aims to capitalize on the power of the primal-dual method for other minor-closed classes and ultimately for arbitrary $\F$ with graphs of bounded treewidth, 
more sophisticated weighting scheme is needed. 
As we already mentioned, this was successfully done by Fiorini, Joret and Pietropaoli~\cite{fiorini2010hitting} for \DIA,
where their primal-dual algorithm is based on  an intricate LP formulation. 
Our primal-dual algorithm  diverges from such tactics, and instead use  the following technical theorem as a guide for the weighting scheme. 
The formal definitions of 
{\sl $c$-outgrowth} and {\sl cluster collection} are  given in Section~\ref{label_unsuspicious}. Intuitively, one may see a {\em $c$-outgrowth} as a $θ_{c}$-minor free subgraph of $G$ with two vertices in common with the rest of the graph. Also, 
a cluster  collection is a collection  of pairwise disjoint connected sets and the {\em capacity} of this collection is the maximum size of its elements.

\begin{theorem}\label{label_quincaillerie}
There is a  function $\newfun{label_quotedthepartinmy}:\Bbb{N}^{2}\to\Bbb{N}$ 
such that, for every two positive integers $c$ and $t,$ there is a uniformly polynomial time algorithm that, given as input  a graph $G,$ outputs one of the following:
\begin{enumerate}
\item a $c$-outgrowth of size at least $c,$
 or 
\item  a {$θ_{c}$-model} $M$ of $G$ of size at most $\funref{label_quotedthepartinmy}(c,t),$ or 
\item a cluster collection  ${\cal C}$ of $G$ of capacity at most $\funref{label_quotedthepartinmy}(c,t)$ such that  $δ(G/{\cal C})\geq t,$ or 
\item a report that $G$ is  $θ_{c}$-minor free.
\end{enumerate}
\end{theorem}
\noindent (By $δ(G)$ we denote the minimum {\em edge-degree}\footnote{The {\em edge-degree} of a vertex $v$ of $G$, denoted by  ${\sf edeg}_{G}(v)$,  is the number of edges that are incident to $v$.} of a vertex in $G$.)

A variant of Theorem~\ref{label_quincaillerie} was originally proved by Joret, Paul, Sau, Saurabh, and Thomass\'{e}~\cite{JoretPSST14hitti} without the capacity condition on a cluster collection in Case 3. 
With a slight modification of their proof, it is not  difficult to excavate the above statement and we provide a  proof in Section~\ref{label_pomposamente}.
It turns out that imposing the capacity condition of Case 3 is crucial for designing a weighting scheme. 

At each iteration, our primal-dual algorithm invokes 
Theorem~\ref{label_quincaillerie}. Depending on the outcome,  the algorithm either runs a {\sl replacer} (defined in Section~\ref{label_vindictively}), that is used in order to reduce the size of a $c$-outgrowth,  
or computes a suitable weighted graph which we call {\sl \layer} (defined in Section~\ref{label_wineglassful}), using a suitable weighting scheme, 
thus reducing the current weight.
In both cases, we convert the current weighted graph $G=(V,E,w)$ into a new weighted graph $G'=(V',E',w')$ on a strictly smaller number of vertices 
so that an $\alpha$-approximate solution for $G'$ implies an $\alpha$-approximate solution for $G$ for some particular value of $\alpha.$

We stress that the replacer is compatible with any approximation ratio in the sense that the optimal weight of a solution is unchanged and every solution after the  replacement 
can be transformed to a solution that is  at least as good. When Theorem~\ref{label_quincaillerie} reports a constant-sized $\theta_c$-model, it is easy to see that  a uniformly weighted \layer\ suffices. 
The gist of Theorem~\ref{label_quincaillerie} is that in the third case that promises a collection of pairwise disjoint {\sl  constant-sized} connected sets. 

Let us first  consider the simplest such case where all connected sets are singletons, namely when $δ(G)≥t$.  It is not difficult to see that, if we consider $t:=6c$ and under the {\sl edge-degree-proportional} weight function, that is for every $v\in V(G)$, $w(v):={\sf edeg}_G(v)$, 
any feasible solution to  \pumpkin\ is a 4-approximate solution. 
For completeness, we include the proof of this observation for  the  general  \WFD\ problem in  \autoref{label_quiriniennes}. 

In the general case  where we have a collection of  pairwise disjoint connected sets, each of size at most $r,$
the critical observation (\autoref{label_semidoncellas}) is that if the contraction of these sets yields a graph of minimum edge-degree at least $t:=8c,$ then a weighting scheme akin to the 
simple case also works. That is, any feasible solution to \pumpkin\ is a $4r$-approximate solution (Section~\ref{label_tyrannically}). 
The overall primal-dual framework is summarized in Section~\ref{label_wineglassful}.

\section{Basic definitions and preliminary results}\label{label_unsuspicious}

We use $\Bbb{N}$ for the set of non-negative integers and $\Bbb{R}_{\geq 0}$ for the set of {non-negative} reals.
Given a set $X$ and two functions $w_1,w_2: X\to\Bbb{R}_{\geq 0},$ we denote by $w_1\pm w_2: X\to\Bbb{R}_{\geq 0}$ the function where for each $x\in X,$ $(w_{1}\pm w_{2})(x)=w_{1}(x)\pm w_{2}(x).$
Given some $r\in \Bbb{N},$ we define $[r]=\{1,\ldots,r\}.$
Given some collection ${\cal A}$ of objects on which the union operation can be defined,
we define $\cupall{\cal A}=\bigcup_{A∈ {\cal A}}A.$
\medskip

All graphs we consider are multigraphs without loops.
We denote a graph by $G=(V,E)$ where $V$ and $E$ are its vertex and edge set respectively. When we deal with vertex-weighted graphs, we denote them by $G=(V,E,w)$ where $w:V(G)\to\Bbb{R}_{\geq 0}$ and we say that $G$ is a $w$-weighted graph.
When considering edge contractions we sum up edge multiplicities of multiple edges that are created during the contraction. However, when a loop appears after a contraction, then we suppress it.
We use $V(G)$ and $E(G)$ for  the vertex set and the edge multiset of $G.$ We also refer to $|V(G)|$ as the {\em size} of $G.$
If $X⊆V(G),$ we denote by $G[X]$  the subgraph of $G$ induced by $X$ and by $G-X$ the graph $G[V(G)\setminus X].$  
We say that $X$ is {\em connected in $G$} if $G[X]$ is connected.
A graph $H$ is a {\em minor} of a graph $G$ if $H$ can be obtained from a subgraph of $G$ after contracting edges.
Given a graph $H,$ we say that  $G$ is {\em $H$-minor free} if $G$ does not contain $H$ as a minor.

 Given an edge $e$ of a graph $G,$ we denote by $m(e)$ its multipliciy and we define $μ(G)=\max\{m(e)\mid e∈ E(G)\}.$
Given a vertex $v∈ V(G),$ we define the edge-degree of $v$ in $G,$ denoted by  $\edeg_{G}(v),$ as the number of edges of $G$ that are incident to $v.$ We denote by $N_{G}(v)$ the set of all neighbors of $v$ in $G$ and 
we call $|N_{G}(v)|$ {\em vertex-degree} of $v$.
We denote by $δ(G)$ the minimum {edge-degree} of a vertex in $G.$

\paragraph{Bonds and pumpkins.} 
Let $G$ be a graph. 
Given two disjoint subsets $X,$ $Y$ of $V(G),$
the  {\em edges crossing} $X$ and $Y$ is the set of edges with one endpoint in $X$ and the other  in $Y.$ The graph $θ_{c}$  is the graph on two vertices joined by $c$ parallel edges ($θ_{c}$ is also known as the {\em $c$-pumpkin}). 
Notice that $θ_{c}$ is a minor of $G$ iff 
$G$ contains two disjoint connected sets  $X$ and $Y$ 
crossed by $c$ edges of $G.$ We call the union $M:=X\cup Y$ {\em $θ_{c}$-model} of $G.$

%

Given a  bipartition $\{V_1,V_2\}$ of $V(G),$ the set of edges  crossing $V_1$ and  $V_2$ is called the \emph{cut} of $\{V_1,V_2\}$ and an edge set 
is a {\em cut} if it is a cut of some vertex bipartition. For disjoint vertex subsets $S,T\subseteq V,$ an \emph{$(S,T)$-cut} is 
a cut of a bipartition $\{V_1,V_2\}$ such that $S\subseteq V_1$ and $T\subseteq V_2.$ Note that, for a connected graph, cuts are precisely the edge sets whose deletion strictly increases the number of connected components.
A {\em block} of a graph $G$ is either an isolated vertex (a vertex of vertex-degree 0) or a bridge or a biconnected component of $G.$

A minimal non-empty cut is known as a {\sl bond} in the literature. We remark that the bonds of $G$ are precisely the circuits of the cographic matroid of $G.$ 
Given a positive integer $c,$ a {\em $c$-bond} of a graph $G$ is any minimal cut of $G$ of size at least $c.$
 Notice that if a $(X,Y)$-cut of a graph $G$ is a bond and $X\cup Y=V(G),$ then both $X$ and $Y$ are connected sets of $G.$
The problem of finding the maximum $c$ for which a graph $G$ contains a  $c$-bond  has been examined both from the approximation~\cite{HaglinV91appro,CarrFLP00streng} and the parameterized point of view~\cite{EtoHK019}.
We next see  how a $c$-bond is related to a $\theta_c$-model.

\begin{observation}
\label{label_rajeunissement}
For every $c\in \Bbb{N}$,   a graph $G$ contains $\theta_c$ as a minor iff it has a $c$-bond. 
\end{observation}
\begin{proof}
Consider a $\theta_c$-model $M=X\cup Y$ where $X$ and $Y$  are two disjoint connected vertex sets of $G,$ with at least $c$ edges crossing $X$ and $Y.$ 
For every vertex $v$ in the same connected component $C$ as $X\cup Y,$ add $v$ to either of $X$ or $Y$ in which $v$ has a neighbor. 
Let $X'$ and $Y'$ be the resulting sets with $X\subseteq X'$ and $Y\subseteq Y',$ and let $F$ be the edge set crossing $X'$ and $Y'.$ 
Clearly $F$ is the cut of $(X', V\setminus X').$ 

To show that $F$ is a $c$-bond, suppose that the cut $F'$ of a bipartition $(A,B)$ of $V$ is a subset of $F.$ 
If  $\emptyset \neq A\cap X'\cap C \subsetneq X'\cap C,$ 
then the cut of $(A,B)$ contains all the edges crossing $X'\cap C \cap A$ and  $(X'\cap C)\setminus A.$ 
Since $X'\cap C$ is connected, there is at least one   edge $e$ crossing $X'\cap C \cap A$ and  $(X'\cap C)\setminus A.$ Because $e$ is an edge whose both endpoints belong to $X'\cap C,$ 
it follows $e\in F'\setminus F,$ a contradiction. Therefore, either $A\cap X'\cap C=\emptyset$ or $A\cap X'\cap C = X'\cap C$ holds, and likewise either $B\cap Y'\cap C=\emptyset$ or 
$B\cap Y'\cap C = Y'\cap C$ holds. It follows that $F'$ is either an empty set or equals $F,$ thus proving that $F$ is a $c$-bond. The backward direction is straightforward 
from that $F$ takes all its edges from a single connected component of $G.$ 
\end{proof}

\paragraph{Treewidth.}

A \emph{tree decomposition} of a graph $G$ is a pair $(T,\chi),$ where $T$ is a tree
and $\chi: V(T)\rightarrow 2^{V(G)}$
 such that:
\begin{enumerate}
\item $\bigcup_{q \in V(T)} \chi(q) = V(G),$
\item for every edge $\{u,v\} \in E,$ there is a $q \in V(T)$ such that $\{u, v\} \subseteq \chi(q),$ and
\item for every vertex $v\in V(G),$ the set $\{t\in V(T)\mid v\in \chi(t)\}$ is connected in $T.$
\end{enumerate}
%
%
The \emph{width} of a  tree decomposition $(T,\chi)$ is $\max\{ |\chi(q)| \mid {q \in V(T)}\} - 1.$ The treewidth of $G$ is the minimum width over all tree decompositions of $G.$

We need the following two observations. Observations (i) and (ii) follow from the combinatorial results of~\cite{LimniosPPT20edged} and~\cite{BodlaenderLTT97onint} respectively.
\begin{proposition}
\label{label_zahlengleichheit}
For every positive integer $c$ and every $θ_c$-minor free multigraph $G,$ it holds that (i) $|E(G)|\leq 2c\cdot |V(G)|$ and (ii) $G$
has treewidth less than $2c.$
\end{proposition}


\paragraph{Covering bonds.} 
Given a set $S\subseteq V(G),$ we say that $S$ is a {\em $c$-bond cover} of $G$ if $G- S$ is $θ_c$-minor free. 
Notice that $S$ is a $c$-bond cover iff $G\setminus S$ does not contain  a $c$-bond. 
Given a weighted graph $G=(V,E,w)$ with $w: V(G)\to \Bbb{R}_{\geq 0},$
a {\em minimum weight $c$-bond cover} of $G$ is a $c$-bond cover
$S$ where the weight of $S,$ defined as $w(S):=\sum_{v\in S}w(v),$ is minimized.
\smallskip

We consider the following optimization problem for every positive integer $c.$

\begin{center} \fbox{
\begin{minipage}{8.4cm}
\noindent \pumpkin\\
\noindent{\sl Input:} a vertex weighted graph $G=(V,E,w).$\\
\noindent{\sl Solution:}  a minimum weight $c$-bond cover of $G.$
\end{minipage}
}
\end{center}
\medskip

The next proposition follows by applying the general algorithmic results of~\cite{BasteST20hitti} on the graph $θ_{c},$ taking into account~\autoref{label_rajeunissement}. 
\begin{proposition} 
\label{label_satisfacerles}
For 
 every two positive integers $c,t$ there exists a uniformly linear-time algorithm that, given a positive integer $c$ and a $w$-weighted $n$-vertex graph $G$ of treewidth at most $t,$ outputs a minimum weight $c$-bond cover of $G.$
\end{proposition}

%
%
%

\paragraph{Cluster collections.}
Let $G$ be a graph.
A {\em cluster collection} of $G$ 
is a non-empty  collection ${\cal C}=\{C_{1},\ldots,C_{r}\}$ of pairwise disjoint non-empty connected subsets of $V(G).$ 
In case $\cupall {\cal C}=V(G)$ we say that ${\cal C}$ is a {\em  cluster partition} of $G.$
The {\em capacity} of a cluster collection ${\cal C}$ is the maximum number of vertices of a cluster in ${\cal C}.$
We use the notation $G/{\cal C}$  for the multigraph obtained from $G[\cupall{\cal C}]$ by  contracting all edges in $G[C_{i}]$ for each $i∈\{1,\ldots,r\}.$
Given a cluster $C\in {\cal C}$ we denote by ${\sf ext}_{\cal C}(C)$ (or simply ${\sf ext}(C)$) the set of edges with one endpoint in $C$ and the other not in $C.$

{\medskip
The next lemma provides (edge) degree conditions that can guarantee  
the existence of a cluster collection $\mathcal{C}$ of a graph $Z$, such that
for every cluster element $C\in\mathcal{C}$ there are at least  $k/c$
edges of $G$ with one endpoint in $C$ and the other to other clusters.
The proof is based on a greedy selection of the elements of $\mathcal{C}$.
We make use of these conditions in the proof of Lemma \ref{label_indistinctement}.
}

\begin{lemma}
\label{label_verwickelten}
Let $Z$ be a multigraph where $μ(Z)<c$ and let $\{A,B\}$ be a partition of $V(Z)$
such that each vertex in $A$ has edge-degree at least $k,$ $B$ is an independent set of $Z,$ and each vertex in $B$ has vertex-degree at least 2.
Then $Z$ has a cluster collection ${\cal C}$ where each vertex in $A$ belongs in exactly one cluster,
no cluster has more than $k+1$ vertices, and 
%
%
each vertex of $G/{\cal C}$ has edge-degree at least $k/c$ in $G/{\cal C}.$

\end{lemma}

\begin{proof}
Let  $a∈ A$ and let $b∈  N_{G}(a)\cap B.$ Observe that $G$ has less than $c$ edges between $a$ and $b$ and at least one edge between $b$
and some vertex in $A\setminus \{a\}.$ Using this fact, observe that we can pick a subset $B_{a}$ of $N_{G}(a)\cap B$ of at most $k$ vertices  such that at least $k/c$ edges of $G$ have exactly one endpoint in $\{a\}∪ B_{a}$ and the other point in $A\setminus a.$
In fact we can greedily pick such a set $B_{a}$ for every $a∈ A$ in a way that if $a\neq a'$ then $B_{a}\cap B_{a'}=\emptyset,$  
and for each $a\in A,$ there are at least $k/c$ edges between $a\cup B_a$ and $\cupall {\cal C}\setminus (\{a\}∪ B_{a}),$ 
where ${\cal C}:=\{a\cup B_a \mid a\in A\}.$ 
Let us begin with $B_a=\emptyset$ for every $a\in A.$ 
Suppose there exists a vertex $a\in A$ with less than $k/c$ edges 
with one endpoint in $a\cup B_a$ and another in $\cupall {\cal C}\setminus (\{a\}∪ B_{a}).$ 
Let $p$ be the number of edges between $a$ and $\cupall {\cal C}\setminus (\{a\}∪ B_{a}).$ 
Due to the edge-degree lower bound on $a,$ there are at least $k-p>0$ edges between $a$ and $B\setminus \cupall {\cal C}.$ 
Update $B_a$ to a minimal subset of $B\setminus \cupall {\cal C}$ so that there are at least $k-p$ edges between $a$ and $B_a.$ 
Now there are at least $p+(k-p)/c\geq k/c$ edges with exactly one endpoint in $\{a\}∪ B_{a}$ and another in $\cupall {\cal C}\setminus (\{a\}∪ B_{a}).$ 
Therefore, we can continue the procedure until the cluster collection  ${\cal C}=\{\{a\}∪ B_{a}\mid a∈ A\}$ at hand has the desired property.
\end{proof}

\paragraph{Detecting and covering $c$-outgrowths.}
Given a graph $G,$ a $c$-outgrowth of $G$ is a triple ${\bf K}=(K,u,v)$ where $u,v$ are distinct vertices of $G,$  $K$ is a component of $G\setminus \{u,v\},$ $N_{G}(V(K))=\{u,v\},$ and the graph, denoted by  $K^{(u,v)},$ obtained from $G[V(K)∪\{u,v\}]$ if we remove all edges with endpoints $u$ and $v$ is $θ_c$-minor free.
The {\em size} of a $c$-outgrowth of $G$ is the size of $K.$ 

We say that a graph $G$ is {\em $c$-reduced} if every $c$-outgrowth  $(K,u,v)$ of $G$ has size at most $c-1.$ 
Clearly, if $G$ is $1$-reduced then $G$ has no $c$-outgrowth for any $c\geq 1.$

\begin{lemma}
\label{label_exorbitancia}
For 
 every two positive integers $c,c',$ there exists a uniformly polynomial-time algorithm that given an $n$-vertex graph $G,$ outputs, if exists, a $c$-outgrowth  $(K,u,v)$ of size $c'.$\end{lemma}

\begin{proof}
The algorithm considers, for every two distinct vertices $u,v\in V(G),$ all the components of $G-\{u,v\}$ of size $c'.$  
For each such triple $(K,u,v)$ the algorithm checks whether 
$(K,u,v)$ is a $c$-outgrowth, that is whether the graph $K^{(u,v)}$ is $\theta_c$-minor-free. 
This can be checked using the algorithm of~\cite{Bodlaender96alinea}, which runs in time  linear in $n.$ 
Notice that if $K^{(u,v)}$ has treewidth $> 2c,$ then it follows from \autoref{label_zahlengleichheit}.ii that  $K^{(u,v)}$  is not  
$\theta_c$-minor-free. In case $K^{(u,v)}$ has treewidth at most $2c,$ then one may
verify whether $(K,u,v)$ is a $c$-outgrowth by using the algorithm \autoref{label_satisfacerles}: just check whether the minimum $c$-bond cover of $K^{(u,v)}$ has size zero.
\end{proof}

\begin{lemma}
\label{label_materialists}
For 
 every  positive integer $c,$ there exists a uniformly linear-time algorithm that given a $c$-outgrowth ${\bf K}=(K,u,v),$ and integer $i\in \{ 0,\ldots,c-1 \},$ and a vertex weighting $w: V(K)\to\Bbb{R},$ outputs a minimum weight subset $S$ of $V(K)$ such that 
$K^{(u, v)} \setminus S$ does not contain any $\theta_{i + 1}$-model $M=X\cup Y$ with $u$ and $v$ in different sets in $\{X,Y\}$
\end{lemma}

\begin{proof}
In order to compute the required set $S$ consider the graph $G$
obtained by $K^{(u,v)}$ after joining $u$ and $v$ with $c$ parallel edges. Then extend the weighting of $w$ to one on $V(G)$ 
by assigning sufficiently large weights to $u$ and $v,$ e.g. $w(u):=w(v):=\sum_{x\in V(K)}w(x).$
Notice that $G$ is $θ_{2c}$-minor free, therefore, by \autoref{label_zahlengleichheit}.ii, it has treewidth at most $4c.$
Then run the algorithm of \autoref{label_satisfacerles} for $c:=c+i+1$ and $t:=2c$ and output a minimum weight set $S$ such that $G\setminus S$ is
$θ_{c+i+1}$-minor free. Notice that, because $K^{(u,v)}$ is  already $θ_{c}$-minor free and because we added $c$ parallel edges, it follows 
that $K^{(u, v)} \setminus S$ does not contain a $\theta_{i + 1}$-model $M=X\cup Y$ with $u$ and $v$ in different sets in $\{X,Y\}$
if and only if $G\setminus S$ is
$θ_{c+i+1}$-minor free.
Therefore the output $S$ is indeed the required set.
\end{proof}

In the statements of \autoref{label_satisfacerles}, \autoref{label_exorbitancia}, and~\autoref{label_materialists}, the term {\em uniformly linear/polynomial} indicates that it is possible to {\sl construct} an algorithm of running time
$O_{q}(n^{z})$ for some universal constant $z$  (where $z=1$ in case of a linear  algorithm)
and where $q$ constructively depends  only on the other constants (that is $c,$ $c',$ or $t$).
We stress that in the statements of our results we did not make any effort to optimize the running time of the claimed polynomial algorithm.

\section{A structural theorem}\label{label_pomposamente}

In this section we prove the following structural result.

\begin{reptheorem}{label_quincaillerie}
There is a  function $\funref{label_quotedthepartinmy}(c,t):\Bbb{N}^{2}\to\Bbb{N}$ 
such that, for every two positive integers $c,t,$ there is a uniformly polynomial time algorithm that, given as input  a graph $G,$ outputs one of the following:
\begin{enumerate}
\item a $c$-outgrowth of size at least $c,$
 or 
\item  a {$θ_{c}$-model} $M$ of $G$ of size at most $\funref{label_quotedthepartinmy}(c,t),$ or 
\item a cluster collection  ${\cal C}$ of $G$ of capacity at most $\funref{label_quotedthepartinmy}(c,t)$ such that  $δ(G/{\cal C})\geq t,$ or \item a report that $G$ is  $θ_{c}$-minor free.
\end{enumerate}
\end{reptheorem}

The proof of Theorem~\ref{label_quincaillerie} is based on the following lemma, whose proof follows closely the construction and some of the arguments of \cite[Lemma 4.10]{JoretPSST14hitti}. In the proof  we present the points that differentiate our proof from the one of \cite[Lemma 4.10]{JoretPSST14hitti}.

\begin{lemma}
\label{label_indistinctement}
There is a  function $\newfun{label_zuwiderliefe}:\Bbb{N}^2\to\Bbb{N}$ such that for every two positive integers $c,t,$  
 there is a uniformly polynomial time  algorithm that, given as input a $1$-reduced graph $G,$ outputs either a {$θ_{c}$-model} $M$ of $G$ of size at most $ \funref{label_zuwiderliefe}(c,t),$ or some cluster collection  ${\cal C}$ of $G$ of capacity at most $ \funref{label_zuwiderliefe}(c,t)$ and such that  $\delta(G/{\cal C})\geq t,$ or a report that $G$ is  $θ_{c}$-minor free.
\end{lemma}

\begin{proof}
Let $\funref{label_zuwiderliefe}(c,t):=\min\{\rred{k^{r}+k+1},\rred{r+r(k-1)k^r}\},$ where $k:=c\cdot t$ and $r:=(2c)^{2c}\cdot k.$ 
We also use $\ell:=\funref{label_zuwiderliefe}(c,t).$\medskip

First of all, we may assume that all edges have multiplicity less than $c,$ as otherwise the endpoints of such an edge would be the desired  {$θ_{c}$-model} of $G.$
Also we may assume that all blocks  of $G$ contain some {$θ_{c}$-model}. If not, we apply the proof to the  updated $G$ taken by removing  all vertices of these blocks that are not  cut-vertices (take into account  that the union of $θ_{c}$-free blocks is also $θ_{c}$-free). 
We refer to this last assumption as {\sl the block assumption}.

Let $W$ be the set of vertices of $G$ with edge-degree at least $k.$
We next consider a maximal collection ${\cal P}$ of pairwise  vertex-disjoint induced subgraphs of $G-W,$ each isomorphic to a multipath on exactly $r$ vertices (a {\em multipath} is a connected graph $P$ that, when we suppress the multiplicity of its edges, gives a path).  
Let ${\cal C}$ be the set of components of $G- (W∪V(\cupall{\cal P})).$ Observe that each graph in ${\cal C}$ has diameter at most $r-1$ 
and maximum edge-degree at most $k-1,$ therefore 
each graph in ${\cal C}$ has at most $\rred{k^r}$ vertices. 
We refer to the subgraphs of $G$ that belong in ${\cal C}$ (resp. ${\cal P}$) as {\em the ${\cal C}$-clusters} (resp. ${\cal P}$-clusters) of $G$ (see~\autoref{label_lognruning}).

\begin{figure}[h]
  \begin{center}
  \includegraphics[scale=.25]{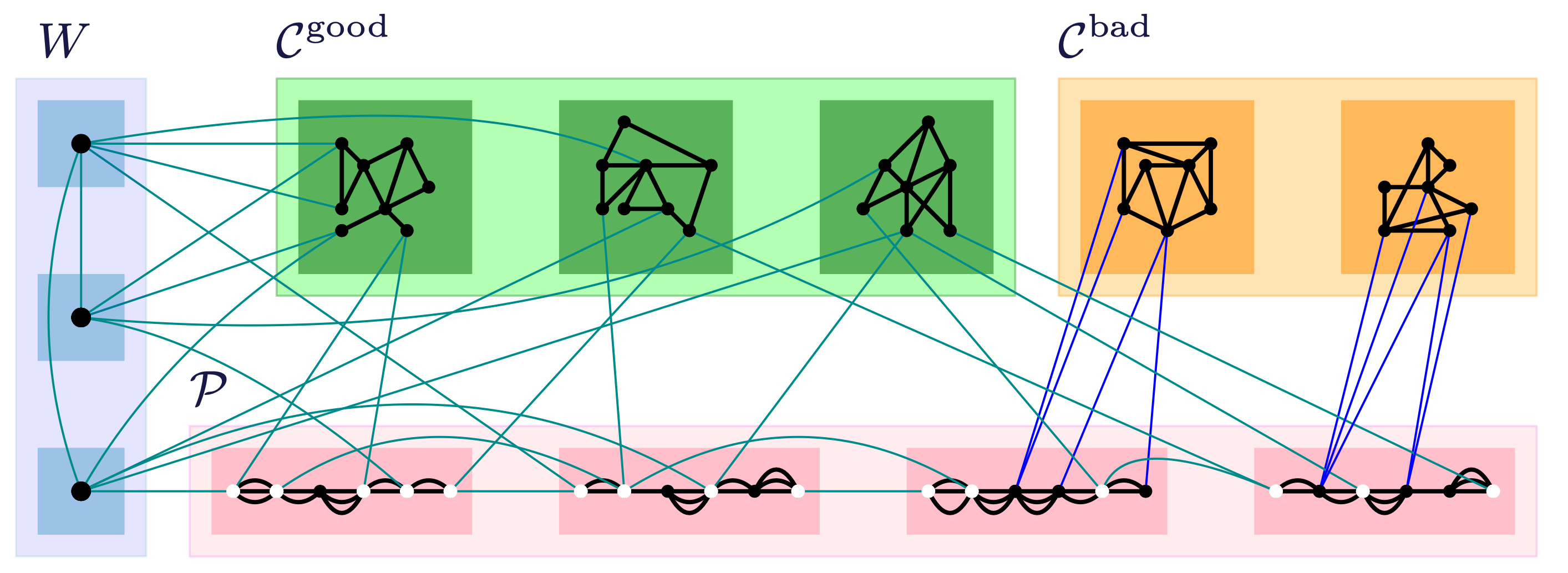}
  \end{center}
\caption{A visualisation of the set $W,$ the ${\cal P}$-clusters, and the ${\cal C}$-clusters in the proof of~\autoref{label_indistinctement}.}
\label{label_lognruning}
\end{figure}

 Let now $K=G/({\cal C}∪{\cal P}),$ i.e., we  contract in $G$ the ${\cal C}$-clusters and the ${\cal P}$-clusters.   Let $v_{C}$ (resp. $v_{P}$)  be the result of the contraction of ${\cal C}$-cluster $C∈ {\cal C}$ (resp. ${\cal P}$-cluster $P∈ {\cal P}$) in $K.$
 We define $V_{c}=\{v_{C}\mid C∈ {\cal C}\},$  $V_{p}=\{v_{P}\mid P∈ {\cal P}\}.$ Observe that $\{W,V_{c},V_{p}\}$ is a partition of $V(K)$ and that $V_{c}$ is an independent set of $K.$
Also, according to {\cite{JoretPSST14hitti}},  either it is possible to find in polynomial time 
some {$θ_{c}$-model} of $G$ of size at most $\max\{k^r+1,r+1,2r,k^r+r\}≤ \ell$
or we may assume that $μ(K)<c.$ We assume the later.
Moreover, we may also assume that $\forall C∈ {\cal C}\ {\sf vdeg}_{K}(v_{C})\geq 1,$ otherwise, by {\sl the block assumption}, a {$θ_{c}$-model} of size $\rred{k^r}≤ \ell$ can be found in $G.$

 We further partition the ${\cal C}$-clusters as follows:
 Let $C∈ {\cal C}.$ If the vertex-degree of $v_{C}$ is one in $K,$ then we say that $C$ is a {\em bad} ${\cal C}$-cluster of $G,$
 otherwise we say that $C$ is a {\em good} ${\cal C}$-cluster of $G.$ This gives rise to the partition  $\{{\cal C}^{\rm good},{\cal C}^{\rm bad}\}$ of ${\cal C}$ by and the corresponding partition  $\{V_{c}^{\rm good},V_{c}^{\rm bad}\}$ of $V_{c}.$  
Recall  that every vertex in $V_{c}^{\rm bad}$ has {\sl exactly one} neighbor in $K.$ Moreover,  this neighbor should 
 be some vertex of $V_{p}.$ Indeed, if this neighbor is a vertex in $W$  then by {\sl the block assumption} 
 we may find some  {$θ_{c}$-model} in $G$ of  size $≤ \rred{k^r+1}≤ \ell$ and we are done.

 Let $v$ be a vertex of some ${\cal P}$-cluster of $G.$
 We say that $v$ is a {\em black} vertex of $P$ if, in $G,$ all its neighbors outside $P$
 are  vertices of bad ${\cal C}$-clusters, i.e., $N_{G}(v)\setminus V(P)⊆  V(\cupall {\cal C}^{\rm bad}).$
 If a vertex of $P$ is not black then  it is {\em white}.
According to  \cite[Claim 4]{JoretPSST14hitti}, either one can find, in polynomial time,
a  {$θ_{c}$-model} of $G$ of size at most $\rred{r+r(k-1)k^r} ≤ \ell$ 
 or there is no ${\cal P}$-cluster that, being a multipath, contains $\rred{(2c)^{2c}}$ consecutive black vertices.
 This implies that we may assume that
  every ${\cal P}$-cluster contains at least $\rred{k}=\frac{r}{(2c)^{2c}}$ white vertices.
  Observe that 
  $${\cal C}'=\{\{w\}\mid w∈ W\}∪\{V(P)\mid v_{P}∈ V_{p}\}\cup\{V(C)\mid v_{C}∈ V_{c}^{\rm good}\}$$
   forms a cluster collection of $G$ of capacity $≤ \rred{k^{r}}$ where we can see $v_P$  (resp. $v_{C}$)
as the result of the contraction of the ${\cal P}$-cluster $V(P)$ (resp. the good ${\cal C}$-cluster $V(C)$). As we have seen above,  every vertex $v_{P}$ has edge-degree at least $\rred{k}$ in $G/{\cal C}'.$
Moreover, less than $c$ edges may exist between each pair of the clusters of ${\cal C}',$ otherwise we obtain a {$θ_{c}$-model} in $G$ of size at most $\rred{2k^{r}}≤  \ell.$ This means that $μ(G/{\cal C}')< c.$ As $Z:=G/{\cal C}'$ satisfies 
the conditions of \autoref{label_verwickelten} for $A:=W∪V_{p}$ and $B:=V_{c}^{\rm good},$ we have that 
there exists a cluster collection ${\cal C}''$ of $Z$ where each vertex $w∈ W\cup V_p$ belongs in exactly one cluster $C_{w},$ no cluster has more than $\rred{k+1}$ vertices, and each vertex of $Z/{\cal C}''$ 
 has edge-degree at least $\rred{k/c}.$ As a last step,
 we 
further merge the clusters in ${\cal C'}$ as indicated by the cluster collection ${\cal C}''$ of $G/{\cal C}'$ (during this merging, the clusters of ${\cal C}'$ that do not correspond to vertices in ${G}/{\cal C}'$ are discarded).
That way we obtain 
a cluster collection ${\cal C}'''$ of $G$ of capacity $≤ \rred{k^{r}+k+1} ≤ \ell$
and where $\delta(G/{\cal C}'')\geq \rred{k/c}=t$
as required.
\end{proof}

We are now in position to give the proof of \autoref{label_quincaillerie}.

\begin{proof}[Proof of \autoref{label_quincaillerie}]
Let $\funref{label_quotedthepartinmy}(c,t):=c\cdot \funref{label_zuwiderliefe}(c,t),$
where $\funref{label_zuwiderliefe}$ is the function of \autoref{label_indistinctement}. By applying \autoref{label_exorbitancia}, we may assume that every $c$-outgrowth of $G$ has size smaller than $c.$ We call such a $c$-outgrowth $(K,u,v),$  of size $<c,$ 
{\em maximal} if there is no other $c$-outgrowth $(K',u',v'),$  of size $<c,$ where $V(K')\cup\{u',v'\}$ is a proper subset of 
$V(K)\cup\{u,v\}.$ Also we call a collection ${\cal K}$ of  $c$-outgrowths of size $<c$ {\em complete} if for every two distinct $(K,u,v),(K',u',v')\in {\cal K},$ $V(K)\cap V(K')=\emptyset.$
Let now ${\cal K}$ be some complete collection of maximal outgrowths 
of $G.$ Notice that ${\cal K}$ can be computed in polynomial time by greedily including in it maximal $c$-outgrowths $(K',u',v')$ of $G,$  of size $<c$ as long as the completeness condition is not violated.

Starting from $G,$ we construct an auxiliary graph $G'$ as follows: for each   $c$-outgrowth ${\bf K}=(K,u,v)\in {\cal K},$ we remove $V(K)$ and 
introduce a new edge whose multiplicity $i_{\bf K}$ is the maximum $i$ such that 
$K^{(u,v)}$ contains a $\theta_{i}$-model $M=X\cup Y$ with $u$ and $v$ in different sets in $\{X,Y\}$ (the multiplicity $i_{\bf K}$ is summed up to the so far multiplicity). The value of $i_{\bf K}$ can be computed using the algorithm of \autoref{label_materialists} because it is equal to  the {minimum $i\in \{ 0,\ldots,c-1 \}$ for which this algorithm  returns the empty set}. Notice that $μ(G')<c,$ otherwise we can find in $G$ a $θ_c$-model of at most $c^2\leq c\cdot \funref{label_zuwiderliefe}(c,t)$ vertices (recall that each edge in $G'$ is either an edge of $G$ or may correspond to a $c$-outgrowth of size $<c$). 
By the maximality and the completeness of the  $c$-outgrowths  of ${\cal K}$ it follows that $G'$ is a $1$-reduced graph. Therefore, we can apply \autoref{label_indistinctement}  on $G'$ and obtain either  a {$θ_{c}$-model} $M'$ of $G'$ of size at most $\funref{label_zuwiderliefe}(c,t)$ or some cluster collection  ${\cal C}'$ of $G'$ of capacity at most $\funref{label_zuwiderliefe}(c,t)$ and where  $\delta(G'/{\cal C}')\geq t.$

If the output of \autoref{label_indistinctement} is a {$θ_{c}$-model} $M'$ of $G'$ of size at most $ \funref{label_zuwiderliefe}(c,t),$ then this gives rise to a {$θ_{c}$-model} $M$ of $G$ of size at most $c\cdot \funref{label_zuwiderliefe}(c,t).$
In case  \autoref{label_indistinctement} gives some cluster collection  ${\cal C}'$ of $G'$ of capacity at most $\funref{label_zuwiderliefe}(c,t)$ then ${\cal C}'$ 
can be straightforwardly  transformed to a cluster collection ${\cal C}$ of $G$ of capacity at most $c\cdot \funref{label_zuwiderliefe}(c,t)$ where $\delta(G/{\cal C})\geq \delta(G'/{\cal C}')\geq t,$ as required.
\end{proof}

\section{The weighting scheme}\label{label_tyrannically}

Let $G$ be a graph and let ${\cal C}$ be a cluster partition of $G$ of capacity at most $r.$ Let also $G$ be an instance of {\sc $c$-Bond Cover} for some positive integer $c.$   
We define the vertex weighting function $w_{\cal C}:V(G) \rightarrow \Bbb{R}_{\geq 0}$ so that if $v∈C \in \calC,$ then 
\begin{eqnarray}
w_{\cal C}(v) & = & \frac{|{\sf ext}(C)|}{|C|}.\label{label_imposibilitado}\end{eqnarray}

When ${\cal C}$ is clear from the context, we simply write $w$ instead $w_{\cal C}.$
The main result of this section is that, with respect to the weight function $w$ in \autoref{label_imposibilitado}, {every} $c$-bond covering  of $G$ is a $4r$-approximation.

\begin{lemma} 
\label{label_semidoncellas}
Let  $c$ be a non negative integer, $G$ be a graph, $r$ be a positive integer, ${\cal C}$ be a connected partition of $G$ of capacity at most $r$ and such that  $\delta(G/{\cal C})≥ 8 c,$ and $w:V(G) \rightarrow \Bbb{R}_{\geq 0}$ be a vertex weighting function as in \mbox{\rm \autoref{label_imposibilitado}}. Then 
for every $c$-bond covering  $X$ of $G,$ it holds that 
\[
\frac{1}{2r}\cdot |E(G/\calC)| ≤ \sum_{v∈ X}w(v) ≤ 2\cdot |E(G/\calC)|.
\]
\end{lemma}

\begin{proof}
For the upper bound, note that
$\sum_{v∈ X}w(v) = \sum_{v∈ X}\frac{|{\sf ext}(C)|}{|C|}\leq\sum_{C∈ {\cal C}}\sum_{v∈ C}\frac{|{\sf ext}(C)|}{|C|}=\sum_{C∈ {\cal C}}|{\sf ext}(C)|=\sum_{x\in V(G / \calC)}\edeg_{G / \calC}(x)=2\cdot |E(G/\calC)|.$

For the lower bound, let $X$ be a $c$-bond covering of $G$ and let  $F=V(G)\setminus X,$ ${\cal C}_{X}=\{C\in{\cal C}\mid C\cap X \neq \emptyset\},$ and ${\cal C}_{F}={\cal C}\setminus{\cal C}_{X}.$
Since $\delta(G / \calC) \geq 8c,$ we obtain that $ |E(G / \calC)|/2\geq 2c\cdot|V(G / \calC)|=2c\cdot |{\cal C}|.$ 
We claim  that $\sum_{C\in \calC_{X}}|{\sf ext}(C)| \geq |E(G / \calC)| / 2.$ Indeed, if this is not the case then, by the fact that $|E(G / \calC)| \leq|E(G[F]/\calC_{F})|+\sum_{C\in \calC_{X}}|{\sf ext}(C)|,$ 
we have that $|E(G[F]/\calC_{F})|> |E(G / \calC)|/2≥ 2c\cdot|\calC|\geq 2c\cdot|\calC_F|$ and this last inequality,  from \autoref{label_zahlengleichheit}.i, gives that
$θ_{c}$ is a minor of  $G/\calC_{F}$ which is a minor of $G[F],$ a contradiction. 
Therefore, since each set in ${\cal C}_{X}$ contains at least one  vertex of $X,$
\[
\sum_{v \in X} w(v)
\geq 
\sum_{C\in \calC_{X}} \frac{|\ext(C)|}{|C|} \geq\frac{1}{r}
\sum_{C\in \calC_{X}} |\ext(C)| \geq \frac{|E(G / \calC)|}{2r},
\]
which proves the lower bound. 
\end{proof}

\section{Replacing outgrowths}\label{label_vindictively}

In this section, we present our $c$-outgrowth replacer. Given a $w$-weighted  graph $G,$ 
we  denote by $\opt_{c}(G)$ (or, for simplicity, $\opt(G)$) the weight of an optimal solution for \pumpkin\  on $G.$
The main result of the section is as follows.

\begin{lemma}
\label{label_empoisonneuse}
For every positive integer $c,$ there is a uniformly polynomial time algorithm that, given a  $w$-weighted graph $G$  and a $c$-outgrowth  ${\bf K} = (K, u, v)$  of $G,$ outputs a weighted graph $G'$ where ${\bf K}$ is replaced by another $c$-outgrowth ${\bf K'} = (K', u, v)$ of size at most $c - 1$ such that 
\begin{enumerate}
\item $\opt(G) = \opt({G'}).$ 
\item Given a $c$-bond cover $S' \subseteq V(G')$ of $G',$ one can construct in polynomial time a $c$-bond cover $S \subseteq V(G)$ of $G$ such that $w(S) \leq w(S').$ 
\end{enumerate}
In particular, an $\alpha$-approximate solution for $G'$ implies an $\alpha$-approximate solution for $G.$ 
\end{lemma}
\begin{proof}
For $i \in \{0, \dots, c - 1 \},$ let $K^{(u,v)}_i$ be the graph obtained from $K^{(u, v)}$ by adding $i$ edges connecting $u$ and $v.$ 
Obviously $K^{(u,v)}_0$ equals $K^{(u,v)}.$ 
Let also $T_i \subseteq V(K)$ be a minimum weight set contained in $V(K)$ such that $K^{(u,v)}_i-T_i$ is $\theta_c$-minor-free, and $w_i = w(T_i).$ 
Note that $T_i \subseteq V(K)$ implies that $T_i$ contains neither $u$ nor $v.$
For example, $T_{c-1}$ is a minimum (internal) vertex cut separating $u$ and $v$ in $K^{(u, v)},$ and $w_{c-1} = w(T_{c-1})$ is finite since there is no edge between $u$ and $v$ in $K^{(u,v)}.$ 
By definition, it holds that 
\[
0= w_0 \leq w_1 \leq \cdots \leq w_{c-1}<\infty,
\]
and by \autoref{label_materialists}, these values can be computed in polynomial time. 
We also remark that $T_j$ is a $c$-bond cover of $K^{(u,v)}_i$ for all $i\leq j.$

We construct the $c$-outgrowth ${\bf K}'=(K',u,v)$ so that $K'^{(u, v)}$ is as follows (see \autoref{label_uninfluenced}). 
\begin{itemize}
\item $V(K'^{(u,v)}) = \{ u, v, x_1, \dots, x_{c-1} \}$ where $K' = \{ x_1, \dots, x_{c-1} \}.$ For each $1 \leq i \leq c-1,$ the weight of $x_i$ is $w_{i}.$ 
\item There are edges $(u, x_{1}), (x_{1}, x_{2}), \dots, (x_{c-2}, x_{c-1}), (x_{c-1}, v).$ Additionally for each $2 \leq i \leq c - 1,$ there is an edge $(x_i,u).$ 
\end{itemize}
\begin{figure}[h]
\vspace{-1mm}
	\begin{center}
 \begin{center}
  \includegraphics[scale=.157]{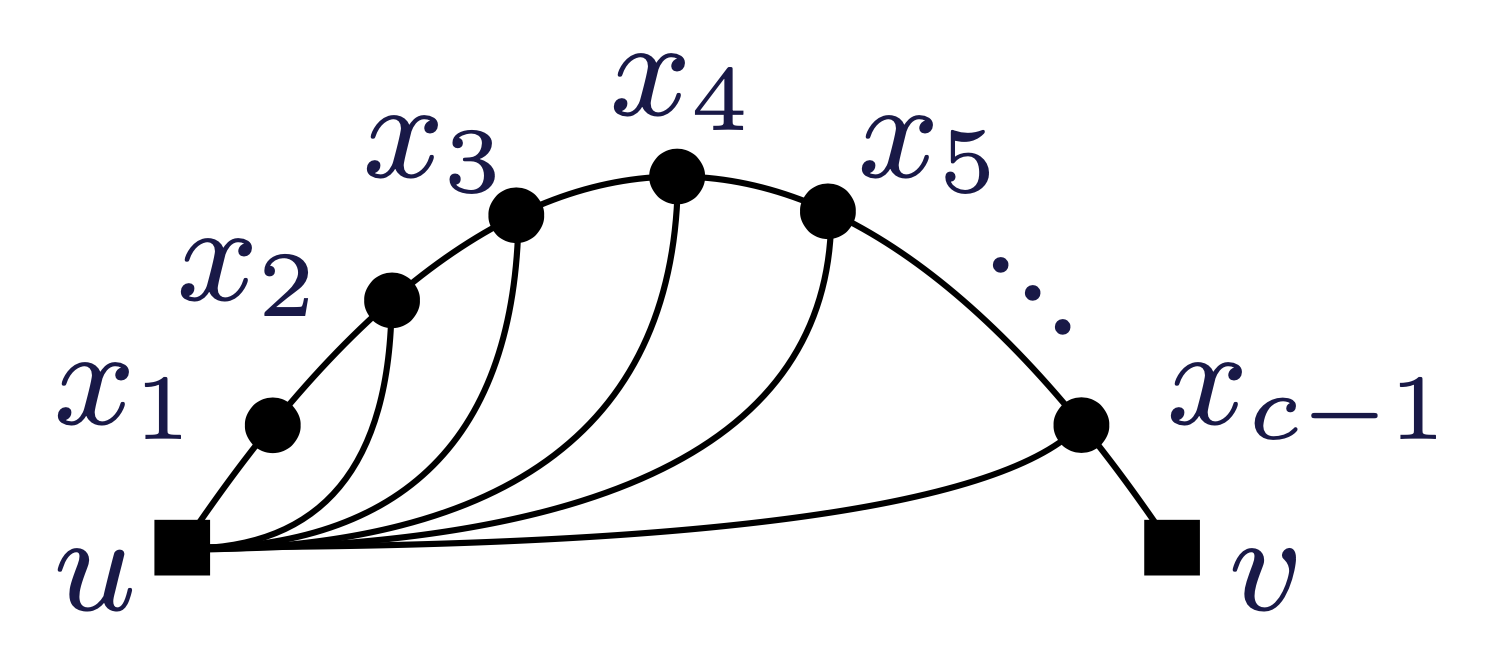}
  \end{center}
	\end{center}
	\caption{The construction of the replacement $c$-outgrowth ${\bf K}'=(K',u,v).$}
	\label{label_uninfluenced}
\end{figure}

We observe that for each $i \in \{0, \dots, c - 1 \},$ the set $\{x_i\}$ is the minimum weight $c$-bond cover of $K'^{(u,v)}_i.$
The next claims are handy.

\begin{claim}
\label{label_satzzeichens}
Let $(K,u,v)$ be a $c$-outgrowth in $G$ and let $M=(X,Y)$ be a minimal $\theta_c$-model in $G.$ If $M$ does  not contain $u,$ then we have $(X\cup Y)\cap V(K)=\emptyset.$ 
Furthermore, if $S$ is a minimal $c$-bond cover of $G$ and if $S$ contains $u$ or $v,$ say $u,$ then $S \cap V(K) =\emptyset.$ 
\end{claim}
\begin{proofofclaim}
Consider a minimal $\theta_c$-model $M = X \cup Y$ of $G,$ and suppose that $M$ does not contain $u.$ 
If $X \cup Y$ contains $v$ (say $v \in X$), then $Y$ is disjoint from $V(K) \cup \{ u, v \}$ because otherwise $Y\subseteq K$ 
and $M'=(X\cap (V(K)\cup \{v\}),Y)$ is a $\theta_c$-model in $K^{(u,v)},$ a contradiction. 
From $Y$ being disjoint from $V(K) \cup \{ u, v \},$ 
we deduce that $M' = (X \setminus V(K)) \cup Y$ is a $\theta_c$-model disjoint from $V(K).$
If $X \cup Y$ contains neither $u$ nor $v,$ it implies that both $X$ and $Y$ are disjoint from $V(K) \cup \{u,v\}.$ 
This proves the first claim. 

To see the second claim, suppose the contrary, i.e. $S\cap V(K)\neq \emptyset.$ 
Then the minimality of $S$ implies that $S\setminus V(K)$ is not a $c$-bond cover of $G$ 
and especially there is a minimal $\theta_c$-model $M=(X,Y)$ in $G- S\setminus V(K).$ 
However, the first claim implies that $M$ is disjoint from $K,$ therefore it is already present in $G-S.$ 
This contradicts that $S$ is a $c$-bond cover of $G,$ completing the proof.  
\end{proofofclaim}

\begin{claim}\label{label_illustrating}
Let $Z$ be a $c$-bond cover of $G-V(K)$ and let $\ell$ be the maximum integer\footnote{If $u$ and $v$ are not connected in $G-(K\cup Z),$ we let $\ell=0.$} 
such that $G-(K\cup Z)$ contains a $\theta_{\ell}$-model with 
$u$ and $v$ in different sets. Then $Z'=Z \cup T_{\ell}$ is a $c$-bond cover of $G.$
\end{claim}
\begin{proofofclaim}
Suppose not, and let $M=(X,Y)$ be a minimal $\theta_c$-model in $G-Z'$ and note that $(X\cup Y)\cap V(K)\neq \emptyset$ 
since otherwise $M$ is also a $\theta_c$-model in $G-(V(K)\cup Z),$ a contradiction.
Due to Claim~\ref{label_satzzeichens} and the minimality of $M,$ we know that $X\cup Y$ must contain both $u$ and $v.$ 
Without loss of generality, let $u\in X.$ 
There are two possible cases.

\medskip
\noindent {Case A. $u, v \in X$ and $Y\subseteq K:$} Note that $G[X\setminus V(K)]$ is connected since otherwise we have the $\theta_c$-model $M$ entirely 
contained in $K^{(u,v)},$ a contradiction. This also implies that $\ell\geq 1.$ Let $P$ be a $(u,v)$-path in $G[X\setminus V(K)].$ 
Because $K^{(u,v)}_1-T_{\ell}$ is a minor of $G[X\cup Y\cup (V(K)\setminus T_{\ell})]$ obtained by replacing $P$ by an edge connecting $u$ and $v,$ 
the graph obtained from $M=(X,Y)$ by replacing $P$ by an edge connecting $u$ and $v$ is a $\theta_c$-model in  $K^{(u,v)}_1-T_{\ell}.$ 
However, $K^{(u,v)}_1-T_{\ell}$ is a minor of $K^{(u,v)}_{\ell}-T_{\ell},$ and  $K^{(u,v)}_{\ell}-T_{\ell}$ is $\theta_c$-minor-free by 
the construction of $T_{\ell},$ a contradiction.

\noindent {Case B. $u \in X$ and $v \in Y$:} Observe that both $(X\cap V(K^{(u, v)}),Y\cap V(K^{(u, v)}))$ and 
$(X\setminus V(K), Y\setminus V(K)$ are a $\theta_i$-model and a $\theta_j$-model for some $i,j\geq 1$ such that $i + j\geq c$ 
and $j\leq \ell<c$ by the choice of $\ell.$ This means that there exists a $\theta_c$-model in $K^{(u,v)}_j-T_{\ell}$ because $K^{(u,v)}_j-T_{\ell}$  is a minor of 
$G[X\cup Y\cup (V(K)\setminus T_{\ell})]$ obtained by contracting $X\setminus V(K)$ and $Y\setminus V(K),$ and 
contracting these sets in  $M$ preserves the property of being a $\theta_c$-model. However, 
$K^{(u,v)}_j-T_{\ell}$ is a minor of $K^{(u,v)}_{\ell}-T_{\ell}$ due to $j\leq \ell,$ which contradict the construction of $T_{\ell}.$ 

%

\medskip
So, in both cases we reach a desired contradiction. That is, $Z'$ is a $c$-bond cover of $G.$ This proves the claim.
\end{proofofclaim}

We begin with proving the second part of the statement. 
Let $G'$ be the graph where $K^{(u, v)}$ is replaced by $K'^{(u, v)}.$ 
It suffices to prove the second statement for an arbitrary {\sl minimal} $c$-bond cover $S'\subseteq V(G')$ of $G'.$ 

First, assume that $S'$ contains $u$ or $v,$ say $u.$ Claim~\ref{label_satzzeichens} 
is applied to $G'$ verbatim with $G \leftarrow G',$ $K \leftarrow K',$ $K^{(u,v)} \leftarrow K'^{(u, v)},$ 
and we deduce that $S'\cap V(K')=\emptyset.$ 
Now we take $S\leftarrow S',$ and let us argue that $S$ is a $c$-bond cover of $G.$ 
Again Claim~\ref{label_satzzeichens} implies that if $G-S$ contains a $\theta_c$-model, then one 
can find one disjoint from $V(K).$ This is not possible because $S=S'$ is a $c$-bond cover of $G-K=G'-K'.$

Secondly, let us assume that $S'\cap \{u,v\}=\emptyset.$ 
Let $\ell$ be the maximum integer such that $G'-(K'\cup S')$ contains a $\theta_{\ell}$-model $M=(X,Y)$ with 
$u$ and $v$ in different sets, say $u\in X$ and $v\in Y.$ 
Clearly $\ell$ is strictly smaller than $c$ because $S'$ is a $c$-bond cover of $G'-K'.$ 
Note that $K'^{(u, v)}_{\ell}$ is obtained from $G'[X\cup Y\cup V(K')]$ by contracting $X$ and $Y.$ 
Because $S'\cap V(K')$ is a $c$-bond cover of $G'[X\cup Y\cup V(K')],$ it is also a $c$-bond cover of $K'^{(u, v)}_{\ell}.$ 
Therefore we have $w_{\ell}\leq w(S'\cap V(K')).$

Let $S = (S' \setminus V(K')) \cup T_{\ell}$ be a vertex set of $G$ and note that $w(S) \leq w(S').$ 
Now applying Claim~\ref{label_illustrating} to $G$ with $Z\leftarrow S'\setminus V(K')$ (as a vertex set of $G$), we 
conclude that $S$ is a $c$-bond cover of $G.$  This proves the second part of the statement, which also 
establishes $\opt(G) \leq \opt({G'})$ in the first part of the statement.

It remains to show $\opt(G) \geq \opt({G'}).$ 
Consider an optimal $c$-bond cover $S$ of $G,$ and let $p$ be the maximum integer such that $G-(K\cup S)$ contains a $\theta_{p}$-model $M=(X,Y)$ with 
$u$ and $v$ in different sets. Again we apply Claim~\ref{label_illustrating} with $G\leftarrow G',$ $Z\leftarrow S\setminus V(K)$ (as a vertex set of $G'$), 
$K\leftarrow K'$ and $T_{\ell}\leftarrow \{x_p\},$ and derive that $(S\setminus V(K))\cup \{x_p\}$ is a $c$-bond cover of $G'.$ 
Lastly, observe that $K^{(u,v)}_{p}$ is a minor of $G[X\cup Y\cup V(K)],$ and 
because $S\cap V(K)$ is a $c$-bond cover of the latter, it is also a $c$-bond cover of the former. 
Therefore, we have $w(S\cap V(K))\geq w_p,$ from which we have $\opt(G)=w(S)\geq w(S\setminus V(K))+w_p=w((S\setminus V(K))\cup \{x_p\})\geq \opt(G').$ 
This finishes the proof.
\end{proof}

\section{The primal-dual approach}\label{label_wineglassful}

We begin the section by formalizing the notions of replacer and \layer. 
A \emph{$c$-outgrowth replacer} (hereinafter \emph{replacer}) is a polynomial-time algorithm which, given a weighted graph $G=(V,E,w)$ 
and a $c$-outgrowth ${\bf K} = (K, u, v)$ of size at least $c,$ outputs a weighted graph $G'=(V',E', w')$ with the following property.
\begin{enumerate}
\item ${\bf K}$ is replaced by another $c$-outgrowth ${\bf K'} = (K', u, v)$ of size at most $c - 1.$
\item $\opt(G) = \opt({G'}).$ 
\item Given a $c$-bond cover $S' \subseteq V(G'),$ one can construct in polynomial time a $c$-bond cover $S \subseteq V(G)$ such that $w(S) \leq w(S').$ 
\end{enumerate}

An \emph{\layer} of a weighted graph $G=(V,E,w)$ is a weighted graph $H=(V,E,w^o)$ such that the following holds.
\begin{itemize}
\item $w^o(v) \leq w(v)$ for every $v\in V,$ 
\item $w^o(v) = w(v)$ for some $v\in V,$ and 
\item $w^o(S) \geq (1/\alpha)\cdot w^o(V)$ for any $c$-bond cover $S\subseteq V$ of $H.$
\end{itemize}
We are now ready to prove our main approximation result.

\begin{reptheorem}{label_resourcefulness}
There is a uniformly polynomial-time algorithm which, given a positive integer $c$ and a weighted graph $G = (V, E, w),$ 
computes a $c$-bond cover of weight at most $\alpha \cdot \opt(G)$ for some $\alpha = \alpha(c).$ 
\end{reptheorem}
\begin{proof}
The algorithm initially sets $G_1=G,$ and iteratively constructs a sequence of weighted graphs $G_i=(V_i,E_i,w_i)$ for $i=0,1,\ldots.$ 
At $i$-th iteration, we run the algorithm $\mathcal A$ of Theorem~\ref{label_quincaillerie} for $t=8c.$ 
Recall that one of the following is the output:
\begin{enumerate}
\item a $c$-outgrowth of size at least $c$ or 
\item  a {$θ_{c}$-model} $M$ of $G$ of size at most $\funref{label_quotedthepartinmy}(c,8c)$ or 
\item a cluster collection  ${\cal C}$ of $G$ of capacity at most $\funref{label_quotedthepartinmy}(c,8c)$ where $δ(G/{\cal C})\geq 8c,$ or
\item a report that $G$ is $\theta_c$-minor-free.
\end{enumerate}

If $\mathcal A$ detects a $c$-outgrowth of size at least $c,$ then we call the algorithm of Lemma~\ref{label_empoisonneuse}, which is clearly a replacer. 
We run the replacer on $G_i$ and set $G_{i+1}$ to be the output of the replacer. 
If {$θ_{c}$-model} $M$ of $G_i$ of size at most $ \funref{label_quotedthepartinmy}(c,t)$ is detected by $\mathcal A,$ 
then let $\epsilon:= \min \{w_i(v):v\in M\}$ and consider the weighted graph\ $H_i=(V_i,E_i,w_i^o)$ with the weight function $w$ as follows:
\begin{equation*}
w_i^o(v) =
\begin{cases}
\epsilon &\text{if } v\in M\\
0&\text{otherwise.}
\end{cases}
\end{equation*}
It is obvious that $H_i$ is an \layer\ with $\alpha=\funref{label_quotedthepartinmy}(c,t).$

In the third case, note that the cluster collection ${\cal C}$ forms a cluster partition of $G_i[\cupall{\cal C}].$  
Consider the weight function $w: \cupall{\cal C} \rightarrow \Bbb{R}_{\geq 0}$ as in \mbox{\rm \autoref{label_imposibilitado}} of $G_i[\cupall{\cal C}].$  
Let $\epsilon:= \min \{w_i(v)/w(v):v\in \cupall{\cal C} \}$ and $H_i=(V_i,E_i,w_i^o)$ be the weighted graph, where
\begin{equation*}
w_i^o(v) =
\begin{cases}
\epsilon\cdot w(v) &\text{if } v\in \cupall{\cal C}\\
0&\text{otherwise.}
\end{cases}
\end{equation*}
Let us verify that $H_i$ is an \layer\ of $G_i$ for $\alpha=4r,$ where $r=\funref{label_quotedthepartinmy}(c,t).$ 
It is straightforward to see that the first two requirement of \layer\ are met due to the choice of $\epsilon.$
To check the last requirement, consider an arbitrary $c$-bond cover $S\subseteq V_i$ of $H_i.$ 
By Lemma~\ref{label_semidoncellas}, it holds that 
\[
\frac{\epsilon}{2r}\cdot |E(G_i[\cupall{\cal C}]/\calC)| ≤ \sum_{v \in S}w_i^o(v) \leq \sum_{v \in V_i}w_i^o(v) ≤ 2\epsilon\cdot |E(G_i[\cupall{\cal C}]/\calC)|,
\]
and therefore,
\[
\sum_{v \in S}w_i^o(v) \geq \frac{\epsilon}{2r}\cdot |E(G_i[\cupall{\cal C}]/\calC)| \geq \frac{1}{4r}\cdot \sum_{v \in V_i}w_i^o(v).
\]
In both the second and the third cases, we set $G_{i+1}$ to be the weighted graph $(V_i,E_i,w_i-w_i^o)$ after removing all vertices of weight zero. 

Finally, if $\mathcal A$ reports that $G_i$ is $θ_{c}$-minor free, then we terminate the iteration. 
Let $G=G_1, G_2,\ldots , G_{\ell}$ be the constructed sequence of weighted graph at the end, with $G_{\ell}$ 
being a $θ_{c}$-minor-free graph. 
Observe that our algorithm  strictly decreases the number of vertices before the $\ell$-th iteration, 
since the second property of an $\alpha$-thin layer implies that at least one $v \in V_i$ has $w_i(v) = w^o_i(v)$ and will not appear in $G_{i+1}$; therefore, $\ell \leq n.$

To establish the main statement, it suffices to show that there is a polynomial-time algorithm 
which produces a $4r$-approximate solution for $G_i$ 
given a $4r$-approximate solution $T_{i+1}$ for $G_{i+1},$ where $r=\funref{label_quotedthepartinmy}(c,t)$ and $t=8c.$ 
By Lemma~\ref{label_empoisonneuse}, this holds if the execution of $\mathcal A$ at $i$-th iteration calls the replacer. 

Suppose that $i$-th iteration produces an \layer\  $H_i=(V_i,E_i,w_i^o),$ 
and recall that every \layer  produced in our algorithm satisfies $\alpha\leq 4r.$   
As $T_{i+1}$  is a $4r$-approximate solution for $G_{i+1},$  we have
\begin{eqnarray}
\opt({G_{i+1}}) & \geq & (1/4r)\cdot w_{i+1}(T_{i+1}),\label{label_konstantinos}
\end{eqnarray} 

\begin{claim}
\label{label_soudainement}
$T_{i}:=T_{i+1}\cup (V_i\setminus V_{i+1})$ is a $4r$-approximate solution for $G_i.$ 
\end{claim}

\begin{proofofclaim}
Let $D_i=V_i\setminus V_{i+1},$ namely 
the vertices deleted from $G_i$ to obtain $G_{i+1}.$
It is obvious that $T_{i+1}\cup D_i$ is a feasible solution for $G_i,$ that is,  a $c$-bond cover of $G_i$ 
because $G_{i+1}- T_{i+1} = G_i- (T_{i+1} \cup D_i)$ and $T_{i+1}$ is a $c$-bond cover of $G_{i+1}.$
Let $Q \subseteq V_i$ be an optimal solution for $G_i.$ Then $Q$ is a feasible solution for $H_i$ and $Q \cap V_{i+1}$ is a feasible solution for $G_{i+1},$ therefore 
\begin{eqnarray}
w_i^o(Q) & \geq & (1/4r)\cdot w_i^o(V_i) \mbox{~and}\label{label_prerequisites}\\
w_{i+1}(Q\cap V_{i+1}) & \geq & \opt({G_{i+1}})\label{label_geometrically},
\end{eqnarray} 
where the inequality~\ref{label_prerequisites} is due to the third requirement of \layer.
Furthermore, it holds that 
\begin{eqnarray}
&  w_i(v)  = w_i^o(v)+w_{i+1}(v) & \mbox{~for each~} v\in V_{i+1} \mbox{~and }\label{label_confederation}
\\  
&   w_i(v)  =  w_i^o(v)~~~~~~~~~~~~~  &  \mbox{~for each~}  v\in D_i. \label{label_unconsciousness}
\end{eqnarray} 
Therefore,
\begin{align*}
w_i(Q)	&= w_i^o(Q) + w_{i+1}(Q\cap V_{i+1}) & \because\eqref{label_confederation},\eqref{label_unconsciousness}\\
			&\geq (1/4r)\cdot w_i^o(V_i) + \opt({G_{i+1}})  &\because \eqref{label_prerequisites},\eqref{label_geometrically}\\
			&\geq  (1/4r)\cdot w_i^o(T_{i+1}\cup D_i) + (1/4r)\cdot w_{i+1}(T_{i+1}) &\because \eqref{label_konstantinos}\\
			&=(1/4r)\cdot (w_i^o(T_{i+1}) +w_{i+1}(T_{i+1}))+  (1/4r)\cdot w_i^o(D_i)\\
			&= (1/4r)\cdot w_i(T_{i+1} \cup D_i)  &\because \eqref{label_confederation},\eqref{label_unconsciousness}
\end{align*}
and the claim follows.
\end{proofofclaim}
\vspace{-3mm}
We inductively obtain a $4r$-approximate solution for $G_{i},$ and finally for the graph $G_1=G.$ This finishes the proof. 
\end{proof}

\section{Discussion}

In this paper we construct a polynomial constant-factor approximation algorithm for the \WFD\ problem in the case
${\cal F}$ is the class of graphs not containing a $c$-bond or, alternatively, the $θ_{c}$-minor free graphs. The constant-factor of our approximation algorithm is a (constructible) function of $c$ and the running time is uniformly polynomial, that is it runs in time $O_{c}(n^{{\cal O}(1)}).$ Our results, in  case $c=2,$ yield a constant-factor approximation for the {\sc Vertex Weighted Feedback Set}. Also, a
constant-factor approximation for \textsc{Diamond Hitting Set} can easily be derived for the case 
where $c=3.$ For this we apply our results on simple graphs and observe that each time a $θ_{3}$-minor-model appears, this model, under the absence of multiple edges, should contain $4$ vertices and therefore is a minor-model of the diamond $K_{4}^{-}$ (that is $K_{4}$ without an edge).\medskip

Certainly the general open question is whether \WFD\  admits a constant-factor approximation for more general instantiations of the  minor-closed class ${\cal F}.$ In this direction, the challenge is to 
use our approach when the graphs in $\F$ have bounded treewidth (or, equivalently, if the minor obstruction of ${\cal F}$ contains some planar graph). For this, one needs to extend the structural 
result of~\autoref{label_quincaillerie} and, based on this to build a replacer as in \autoref{label_empoisonneuse}.

Given an $r\in \Bbb{N},$ an {\em $r$-protrustion} of a graph $G$ is a set $X\subseteq V(G)$ such that 
$G[X]$ has treewidth at most $t$ and $|\partial_{G}(X)|\leq t,$ were $\partial_{G}(X)$ is the set of vertices of $X$ that are incident to edges not in $G[X].$ We conjecture that a possible extension of~\autoref{label_quincaillerie}
might be the following.

\begin{conjecture}
\label{label_infinitesimal}
There are functions $\newfun{label_cristianesmo}:\Bbb{N}^{2}\to\Bbb{N}$ and 
 $\newfun{label_inconsistent}:\Bbb{N}^{3}\to\Bbb{N}$ 
such that, for every  $h$-vertex planar graph $H$ and every two positive integers $t,p,$ there is a uniformly polynomial time algorithm that, given as input  a graph $G,$ outputs one of the following:
\begin{enumerate}
\item an  $\funref{label_cristianesmo}(h,t)$-protrusion $X$ of size at least $p,$
 or 
\item  a minor-model of $H$ of  size at most $\funref{label_inconsistent}(h,t,p),$ or 
\item a cluster collection  ${\cal C}$ of $G$ of capacity at most $\funref{label_inconsistent}(h,t,p)$ such that  $δ(G/{\cal C})\geq t,$ or 
\item a report that $G$ is  $H$-minor free.
\end{enumerate}
\end{conjecture}

Given a proof of some suitable version of~\autoref{label_infinitesimal} at hand,  cases 1,2, and 3 above can be treated using the method proposed in this paper. In the first case, we need to find a {\sl weighted prorusion replacer} that can replace, in the weighed graph  $G=(V,E,w),$ the subgraph $G[X]$ by another one (glued on the same boundary) and create a new weighted graph $G'=(V',E',w')$  so that an optimal solution has the same weight in both instances. In our case, the role of a protrusion is played by the $c$-outgrowth, where $X$ is the vertex set of $K^{(u,v)}$ that has treewidth at most $2c$ and $|\partial_{G}(X)|\leq 2,$ i.e., $V(K^{(u,v)})$ is a $2c$-protrusion of $G.$ In the case of $θ_{c},$ the replacer is given in~\autoref{label_empoisonneuse}. The existence of such a replacer in the general case
is wide open, first because the boundary $\partial_{G}(X)$ has bigger size (depending on $h$ but perhaps also on $t$) and second, and most important, because we now must deal with {\sl weights} which does not permit us to use any protrusion replacement machinery such as the one used in~\cite{fomin2012planar,FominLMPS16} 
{unweighted version of the problem  (based on the, so called, {\sf FII}-property~\cite{BodlaenderFLPST16metak} for more details).}

We believe that a possible way to prove \autoref{label_infinitesimal} is to use as departure the proof of the main combinatorial result {in~\cite{BatenburgHJR19atigh}.}
However, in our opinion, the most challenging 
step is to design a weighted protrusion replacer  (or, on the negative side, to provide instantiations of $H$ where such a replacer does not exist). As such a replacer needs to work on the presence of weights, we suggest that its design might use techniques related to  mimicking networks technology
(see~\cite{HagerupKNR98chara,KrauthgamerR13mimic}).
%
\medskip

Finally, since our algorithm is based on the primal-dual framework and proceeds by constructing suitable weights for the second and third case  where every feasible solution is ${\cal O}(1)$-approximate, 
one can ask whether it is possible to {\sl bypass} the need for a replacer and construct suitable weights for the first case. 
Indeed, the previous approximation algorithms for \FVS~\cite{BafnaBF99,BECKER1996,ChudakGHW98} designed suitable weights even for the case 1
where every {\em minimal} solution is ${\cal O}(1)$-approximate. 
(And used the additional ``reverse delete'' step at the end to ensure that the final solution remains minimal, for every weighted graph constructed.) 
In Appendix~\ref{label_constraineth}, we show that such weights {\em cannot exist} for a simple planar graph $H,$ which suggests that replacers are inherently needed for this class of algorithms for \WFD.

\remove{
\bibliographystyle{abbrv}
\bibliography{wversion}
}

\removed{

}

\newpage 
\appendix
\section{The case of big minimum edge-degree}\label{label_quiriniennes}

It is known that if ${\cal F}$ is a non-trivial minor-closed graph class, then there is some $d_{\cal F}$ such that for every (multi) graph in ${\cal F}$ it holds that $|E(G)|\leq  μ(G)\cdot d_{\cal F}\cdot  |V(G)|$ (see e.g.~\cite{Thomason01thee}).

\begin{lemma}\label{label_falseggiando}
Let $G$ be an instance of {\rm \FD} where $μ(G)\leq m$ and   $\delta(G)\geq 3m\cdot d_{\cal F}.$ Then   for every solution $X$ of {\rm \FD}  on $G,$ it holds that 
\[
|E(G)| \leq 2\sum_{v\in X}\edeg_{G}(v)\leq 4|E(G)|.
\]
\end{lemma}
\begin{proof}
We prove the first inequality as the second one is obvious.
Let $F:=G-X$ , $d:=d_{\cal F},$ and $b:=\delta(G)\geq 3md.$ Observe that 
\begin{align*}
\abs{E(X,F)}	&= \sum_{v\in F} \edeg_{G}(v) - \sum_{v\in F} \edeg_F(v) \nonumber\\
			&\geq b\cdot \abs{V(F)} - 2\abs{E(F)} \nonumber\\
			&\geq b\cdot \abs{V(F)} - 2md\cdot \abs{V(F)} &\because \abs{E(F)}\leq md\cdot \abs{V(F)} \nonumber\\
			&\geq md \cdot \abs{V(F)} &\because b\geq 3md \nonumber\\
			& \geq \abs{E(F)}
\end{align*}

Therefore, we deduce 
\begin{align*}
\sum_{v\in X}\edeg_{G}(v)		&=2\abs{E(X)}+\abs{E(X,F)}\\
					&= \frac{1}{2}	 \abs{E(X)}+ \frac{1}{2}\abs{E(X,F)} +  \frac{1}{2}\abs{E(X,F)}	\\
					&\geq \frac{1}{2}\big(	 \abs{E(X)}+\abs{E(X,F)} + \abs{E(F)}	\big) \\
					&=\frac{1}{2}\cdot \abs{E(G)}\\
					&=\frac{1}{4}\cdot \sum_{v\in V(G)} \edeg_{G}(v)\\
					& = \frac{1}{2}|E(G)|
\end{align*}
as claimed.
\end{proof}

In the special case where ${\cal F}$ is the class of $θ_{c}$-free graphs, we know, from \autoref{label_zahlengleichheit}.i, that every simple $θ_{c}$-minor free graph $G$ satisfies $E(G)\leq 2c\cdot |V(G)|.$ As for such graphs $μ(G)< c,$ we can apply the above lemma for $d_{\cal F}=2.$
Therefore, given an instance $G$ of (unweighted) \upumpkin\ where $\delta(G)\geq 6c,$ 
every feasible solution $X$ is $4$-approximation. 
This implies that if we set $w(v):=\edeg_{G}(v)$ for every $v\in V,$ every feasible solution to  \pumpkin\ is a 4-approximate solution.

\section{Necessity of a replacer}

\label{label_constraineth}

The previous approximation algorithms for \FVS, at least the ones stated in~\cite{BafnaBF99,BECKER1996,ChudakGHW98}, did not use the notion of a {\em replacer} that we introduce here. 
Instead, these algorithms use the primal-dual framework by {\em always} finding suitable weightings $w : V(G) \to \R_{\geq 0}$ of the current graph $G$ such that {\em any minimal} feedback vertex set (i.e., cycle cover) for $G$ is a $2$-approximate solution with respect to $w.$ (In this section, minimality is defined with respect to set inclusion.)
The existence of such a weighting $w$ relies on the fact that, for any induced path $P$ of $G,$ any minimal feedback vertex set has at most one internal vertex from $P.$ 
We prove the following lemma revealing that such weights cannot exist for $H$-minor cover for some planar graph $H.$ (In fact, every graph $H$ that satisfies the mild technical conditions (1) and (2) in the proof.)

\begin{lemma}
\label{label_chenverbindung}
There exists a fixed planar graph $H$ such that the following holds. 
For infinitely many values of $n,$ there exists an $n$-vertex graph $G$ such that 
for any weighting $w : V(G) \to \R_{\geq 0},$ there exists a minimal $H$-minor cover $S \subseteq V(G)$ of $G$ such that $w(S) = \Omega(n) \cdot \opt(G).$ \end{lemma}

\begin{proof}
Let $H$ be any 2-connected planar graph with treewidth at least five (for this, one may just take the $(5 \times 5)$-grid). 
Choose $s, t \in V(H)$ such that $\{s, t\} \in E(H).$ 
Let $k$ be any {parameter} bigger than five. The graph $G$ is constructed as follows.
\begin{itemize}
\item Let $R = \{ 0, 1, 2, 3 \} \times [k]$ and $V(G) = V(H) \cup R.$ So $n = 4k + |V(H)|.$ 
\item $E(G)$ consists of the following edges ($\oplus$ and $\ominus$ denotes the addition and subtraction modulo 4 respectively).
\begin{itemize}
\item the edges in $E(H) \setminus \{ \{s, t\} \}.$ 
\item $s$ is adjacent to $(i, 1)$ and $t$ is adjacent to $(i, k)$ for each $i \in \{0, 1, 2, 3 \}.$ 
\item $G[R]$ is the $(4 \times k)$-grid. Formally, for $1 \leq j \leq k$ and $0 \leq i < 3,$ $E(G)$ contains the edge  $\{(i, j), (i \oplus 1, j)\}$. Moreover, for $1 \leq j < k$ and $0 \leq i <3,$ $E(G)$ contains the edges $\{(i, j), (i, j+1)\}$. 
\end{itemize}

\end{itemize}

Let $Z:=G[\{s,t\}\cup R].$ We need the following two properties of $H$ (we remark that the lemma is true for any $H$ that satisfies them): 
\begin{itemize}
\item[(1)] $H$ is $2$-vertex connected and  
\item[(2)] $H$ is not a minor of $Z.$
\end{itemize}
For the above, it is easy to verify that $Z=G[\{s,t\}\cup R]$
has treewidth at most four, therefore it cannot  contain as a minor the graph $H$ that has treewidth at least five.\medskip

For $1 \leq j \leq k,$ let $Y_j = \{ (0, j), (1, j), (2, j), (3, j) \}$   and notice that $Y_{i}$ is a {\em minimal $s$-$t$ separator} of $R$ (that is $s$ and $t$ are in different connected components of $R-S$ and none of its subsets has this property). In the proof of the following claim, by 
$H$-{\em minor-model} we mean a cluster collection ${\cal M}=\{M_{v}\mid v\in V(H)\}$
of $G$ such that $G/{\cal M}$ is equal to $H$ with possibly some additional edges.

\begin{claim}
If $S$ is a minimal $s$-$t$ separator of $Z$, then 
$S$  is a $H$-minor cover of $G.$
\end{claim}
\begin{proofofclaim}
Once $S$ is removed from $G,$ there is no $s$-$t$ path in $Z.$ 
Let $A \subseteq R$ be the vertices connected to $s$ and $B \subseteq R$ be the vertices connected to $t$ (so $(A, B,S)$ is a partition of $R$). 
Assume towards contradiction that there exists a $H$-minor-model ${\cal M}=\{M_{v}\mid v\in V(H)\}$ of $G'=G - S.$ 
We consider the following cases.
\begin{itemize}
\item If there exists some $v \in V(H)$ such that $M_v \subseteq A,$ the 2-vertex-connectivity of $H$ ensures that there is no $w \in V(H)$ such that $M_w \subseteq V(G-S) \setminus (A \cup \{ s \}),$
since $s$ is a singleton vertex cut separating $A$ and $V(G-S) \setminus (A \cup \{ s \}).$ By letting $M_v \leftarrow M_v \cap (A \cup \{ s \})$ for each $v \in V(H),$ this implies that there exists a $H$-model contained in $A \cup \{ s \},$ which contradicts the fact $H$ is not a minor of $G[\{ s \} \cup R].$ Similarly, there is no $v \in V(H)$ such that $M_v \subseteq B.$ 
\item Then there is no $v \in V(H)$ such that $M_v \subseteq A$ or $M_v \subseteq B.$ In that case, letting $M_v \leftarrow M_v \setminus (A \cup B)$ for each $v \in V(H)$ gives a $H$-model of $G - S$ entirely contained in $V(G - S) \setminus (A \cup B).$ But note that 
$G - (S \cup A \cup B) = G-R$ is exactly $H$ minus one edge, so contradiction. 
\end{itemize}
\vspace{-8.5mm}
\end{proofofclaim}
\vspace{-2.5mm}
On the other hand, we define, for $0 \leq i < 3,$  the sets $$X_i = (\{ i \ominus 1 \} \times \{ 2, \dots, k - 1 \} ) \cup (\{ i \oplus 1 \} \times \{ 2, \dots, k - 1 \} ) \cup \{ (i, k), (i \oplus 2, 1) \}.$$ 
By construction, in $Z,$ each $X_i$ is a minimal $s$-$t$ separator. 
Therefore, $X_i$ is a minimal $H$-minor-cover of $G.$ 

For any $w : V(G) \to \R_{\geq 0},$ note that $\sum_{j = 1}^k w(Y_j) \leq \frac{1}{2} \sum_{i = 0}^3 w(X_i).$ Indeed, for each $(i, j) \in R,$ 
$(i,j)$ appears at least twice in $\{X_0,X_1,X_2,X_3\}$ and exactly once in $\{Y_j:1 \leq j \leq k\}.$
Since each $Y_j$ is a $H$-minor cover, it holds that $$k\cdot  \opt^{(w)}(G) \leq \sum_{j=1}^k w(Y_j) \leq \frac{1}{2} \sum_{i\in\{0,1,2,3\}} w(X_i).$$ 
(Here we denote by $\opt^{(w)}(G)$ the optimum value attained with respect to the weighting $w.$)  
Therefore, for some $i,$ $X_i$ is a minimal $H$-cover whose weight is at least $\frac{k}{2}\cdot  \opt^{(w)}(G).$ As $|V(G)|=25+4k,$ the 
lemma follows.
\end{proof}
We stress that in the above proof, the treewidth of $G$ is always equal to five (that is the treewidth of $H$). This means that one may not expect a better behavior than the one testified by \autoref{label_chenverbindung}, even when restricted to graphs of fixed treewidth. Notice also that the graph $G$ is not planar (while the graph $Z$ is). Any construction of a planar $G,$ as in \autoref{label_chenverbindung}, seems to be an interesting challenge.

\end{document}